\documentclass{amsart}
\usepackage[dvips]{epsfig}
\usepackage{graphicx}
\usepackage{latexsym}
\usepackage{amsmath}
\usepackage{amsthm}
\usepackage{amssymb}
\RequirePackage[colorlinks=true,citecolor=blue,urlcolor=blue]{hyperref}
\usepackage{scrextend}
\usepackage{hyperref}
\usepackage{multirow}
\usepackage{booktabs}
\usepackage{eurosym}
\usepackage{comment}

\usepackage{xcolor}

\usepackage{caption}
\usepackage{subcaption}
\usepackage{float}

\usepackage[applemac]{inputenc}

\newtheorem{theorem}{Theorem}
\newtheorem{assumption}[theorem]{Assumption}
\newtheorem{lemma}[theorem]{Lemma}

\newtheorem{corollary}[theorem]{Corollary}

\newtheorem{proposition}[theorem]{Proposition}
\newtheorem{remark}[theorem]{Remark}

\newcommand{\vip}{\vskip.2cm}

\newcommand{\E}{\mathbb{E}}
\newcommand{\COMMENTAIRE}[1]{}

\begin{document}

\title[]{Electricity intraday price modeling with marked Hawkes processes}

\author{Thomas Deschatre \and Pierre Gruet}

\address{Thomas Deschatre, EDF Lab Paris-Saclay and FiMe, Laboratoire de Finance des March\'es de l'Energie, 91120 Palaiseau, France}
\email{thomas-t.deschatre@edf.fr}
\address{Pierre Gruet, EDF Lab Paris-Saclay and FiMe, Laboratoire de Finance des March\'es de l'Energie, 91120 Palaiseau, France}
\email{pierre.gruet@edf.fr}

\begin{abstract} 
We consider a 2-dimensional marked Hawkes process with increasing baseline intensity in order to model prices on electricity intraday markets. This model allows to represent different empirical facts such as increasing market activity, random jump sizes but above all microstructure noise through the signature plot. This last feature is of particular importance for practitioners and has not yet been modeled on those particular markets. We provide analytic formulas for first and second moments and for the signature plot, extending the classic results of Bacry et al.~\cite{Bacry2013a} in the context of Hawkes processes with random jump sizes and time dependent baseline intensity. The tractable model we propose is estimated on German data and seems to fit the data well. We also provide a result about the convergence of the price process to a Brownian motion with increasing volatility at macroscopic scales, highlighting the Samuelson effect.
\end{abstract}

\maketitle

\textbf{Mathematics Subject Classification (2020)}: 60G55, 62M10, 91G30

\textbf{Keywords}: Electricity intraday prices, Microstructure noise, Hawkes processes, High-frequency statistics

\section{Introduction}

In Europe, The electricity intraday markets are quickly developing as they meet increasing trading needs from the actors in the field: according to EPEX SPOT, which runs those markets in Western Europe, the German yearly exchanged volume has grown from 2~TWh in 2008 to 54~TWh in 2019. To a large extent, this development is caused by the increasing share of renewable power plants in the energy system: the output of such plants is hard to forecast (one may refer to Giebel et al.~\cite{Giebel2011} for a thorough review on this topic), so that their owners face imbalances between their production and the consumption in their perimeter. For each delivery hour (24 in a day), an auction called ``day-ahead market'' is held at noon the day before: it reflects the supply--demand equilibrium that holds, based on the anticipations of market agents. Yet, those agents can also suffer from random events happening between the day-ahead auction and the electricity delivery, so that they need to rebalance their positions before the delivery. Electricity intraday markets provide them with a way to trade deliveries of electricity for the upcoming hours in order to reduce their imbalances. As regulation varies across market places, the intraday markets may change a bit from one country to another. They however share common features: they offer the possibility to exchange power for the next, say, 9 to 32 hours (in Germany, for instance). Three hours after the day-ahead market is cleared, the intraday market for the next day opens. Then trades occur as in the regular order book markets that are well known to finance practitioners. It is possible to trade up to 5 minutes before the physical delivery of electricity actually occurs. The traded products are much standardized: for instance in Germany, one may exchange hourly, half-hourly and quarter-hourly contracts for, respectively, the 24 hours, the 48 half-hours and the 96 quarters of an hour that exist each day.\par \medskip
Because the electricity markets expanded quite lately compared to more conventional financial markets, practitioners in the former already benefit from the modeling tools developed by academic researchers and practitioners in the latter. Yet the nature of the underlying, which is not storable at large scale, has motivated the development of some modeling artifacts that are specific to the field. Kiesel and Paraschiv~\cite{Kiesel2017} look at the fundamental drivers of intraday prices: they can understand how the intraday prices react to, for instance, errors in the renewable production forecast. In order to gain some more insight about the endogeneous mechanisms driving the electricity intraday markets, Narajewski and Ziel~\cite{Narajewski2019} examine the transaction arrival times in this market to understand the behavior of market participants and to describe the increase in the transaction rhythm as maturity gets closer, which is a largely observed stylized fact in the intraday market. A\"id et al.~\cite{Aid2016} determine the optimal strategy of a trader aiming at reducing her imbalance thanks to the intraday market, additionally featuring market impact leading to additional transaction costs. 
Recently, Hawkes processes have proved to be relevant modeling tools in electricity intraday markets. For an extensive review on those processes, the interested reader may refer to Bacry et al.~\cite{Bacry2015}. They are well known to finance practitioners as they allow to represent self-exciting behaviors and clustering of events; for instance Gao et al.~\cite{Gao2018} use them to model the trades arrival in dark pool trading. Back to the context of electricity intraday markets, Favetto~\cite{Favetto2019} fits an univariate Hawkes process to describe the market activity of a product in the intraday market. Graf von Luckner and Kiesel~\cite{Graf2020} work with a 2-dimensional Hawkes process in the intraday market.\par \medskip
Our model shares a lot of features with the model of Graf von Luckner and Kiesel~\cite{Graf2020}: we also use a 2-dimensional Hawkes process with exponential baseline function and exponential excitation functions to define the positive and the negative jumps of prices in the intraday market. Within this model, the occurrence of positive jumps has an influence on the intensity of negative jumps. The same holds for the occurrence of negative jumps, which have an influence on the intensity of positive jumps. Whereas Graf von Luckner and Kiesel~\cite{Graf2020} focus on the statistical selection of models to describe the occurrence of jumps at best, we choose to go one step further on the modeling path by also selecting the distribution of the price after a thorough statistical analysis, and then we define the price on the intraday market as being the sum of positive and negative jumps triggered by the 2-dimensional marked Hawkes process. Being able to model the price and analyse its theoretical features is a great improvement for the practitioner. A very interesting empirical fact captured by our model is the presence of microstructure noise. We generalize the signature plot formulas of Bacry et al.~\cite{Bacry2013a} to the case where the baseline intensity is not constant and the Hawkes processes are marked. We also extend the results of Bacry et al.~\cite{Bacry2013b} concerning limit theorems at macroscopic scale to our framework. To our knowledge, we are proposing the first price model for electricity intraday prices featuring microstructure effects, with a focus on sticking to the actual signature plot, which represents the realized volatility of prices against the sampling frequency. We care for the statistical properties of the price process: the computation of the first and second moments allows to compare the theoretical signature plots to the ones we compute on the market. By doing so, we provide some evidence of patterns that are well known in classic financial markets, and we reproduce them within the model. This is especially interesting as the intraday market of electricity is rather uncommon in the landscape of financial markets. We also examine the behavior of prices in a low-frequency asymptotics.\par
\medskip
The outline of the paper is organized as follows. In Section~\ref{sec:stylizedfact}, we perform a statistical analysis to show some evidence of the increase of the intensity of jumps over the trading sessions. We also study the distribution of jumps sizes, and we look at the signature plot computed on market data. In Section~\ref{sec:model}, we explain how the features we examine in Section~\ref{sec:stylizedfact} translate into modeling choices. We also state the theoretical moments of the price and the signature plot formula within our model. We then estimate the model on German data and we study the adequacy of the model to the data. 
Section~\ref{sec:macro} examines the behavior of prices at the macroscopic scale and the representation of the Samuelson effect, according to which the prices tend to be more volatile as time to maturity decreases. Section~\ref{sec:conclusion} contains the conclusion, and proofs of our propositions are deferred to Sections~\ref{proof:sectionmodel} and \ref{proof:sectionmacro}.

\section{Empirical stylized facts}
\label{sec:stylizedfact}

In this section, we first describe the dataset and some empirical stylized facts which we want to reproduce within our model. These facts are mainly related to the jumps one may observe in the intraday prices: we will estimate their arrival intensity and comment on the distributions that may fit the sizes of the positive and negative price jumps.
Then we will identify the presence of microstructure noise through a signature plot of electricity intraday prices. These features are of particular interest from the perspective of a practitioner, as they provide guidance for modeling in order to trade on the intraday market.

\subsection{Description of the dataset}

We consider German electricity intraday mid-prices (arithmetic average between bid and ask prices) between July and September 2017 for products with a delivery period of one hour. The mid-prices are built using order book data from EPEX Spot, which is the organizer of the short-term market in Western Europe.
There is one time series per trading session and delivery product. Let us recall that one trading session starts at 3~pm the day before maturity and ends 5 minutes before maturity for hourly products, that are the only ones we consider here. Therefore, the duration of observation changes from a maturity to another. Consider one trading session: each time an order is inserted, we update the mid-price. Prices are then given with a precision up to the millisecond. To simplify the different numerical computations, we build the prices data set with a time step of one second, which is enough for our study. If there are one or several price changes over one second, the average volume weighted price is taken for that second. One hour before delivery, cross-border trading is not possible anymore. Thirty minutes before delivery, transactions are only possible into each of the four control areas in Germany and not across them, meaning the bidding zone then becomes smaller than before entering this 30 minutes delay. Those changes in the market can have an impact on the liquidity and then in the price behavior.
For this reason, we only consider data up to one hour before maturity. Furthermore, as the price is not defined at the beginning of the session (at least one buy and one sell order are needed), the price we register at 3~pm is the first existing price (which is thus defined in the future). This does not affect the results as most of them concern the returns and times of price changes. Figure~\ref{fig:prices_ts} represents prices for deliveries beginning at 18h, 19h and 20h on July 11\textsuperscript{th}, 2017 and on August 30\textsuperscript{th}, 2017. In the rest of the paper, we focus on those three maturities but only for a better visualization and understanding: results are robust to the choice of the maturity (with a delivery period of one hour). 

\begin{figure}
    \centering
     \centering
    \begin{subfigure}{0.9\textwidth}
    \centering
    \includegraphics[width=\textwidth]{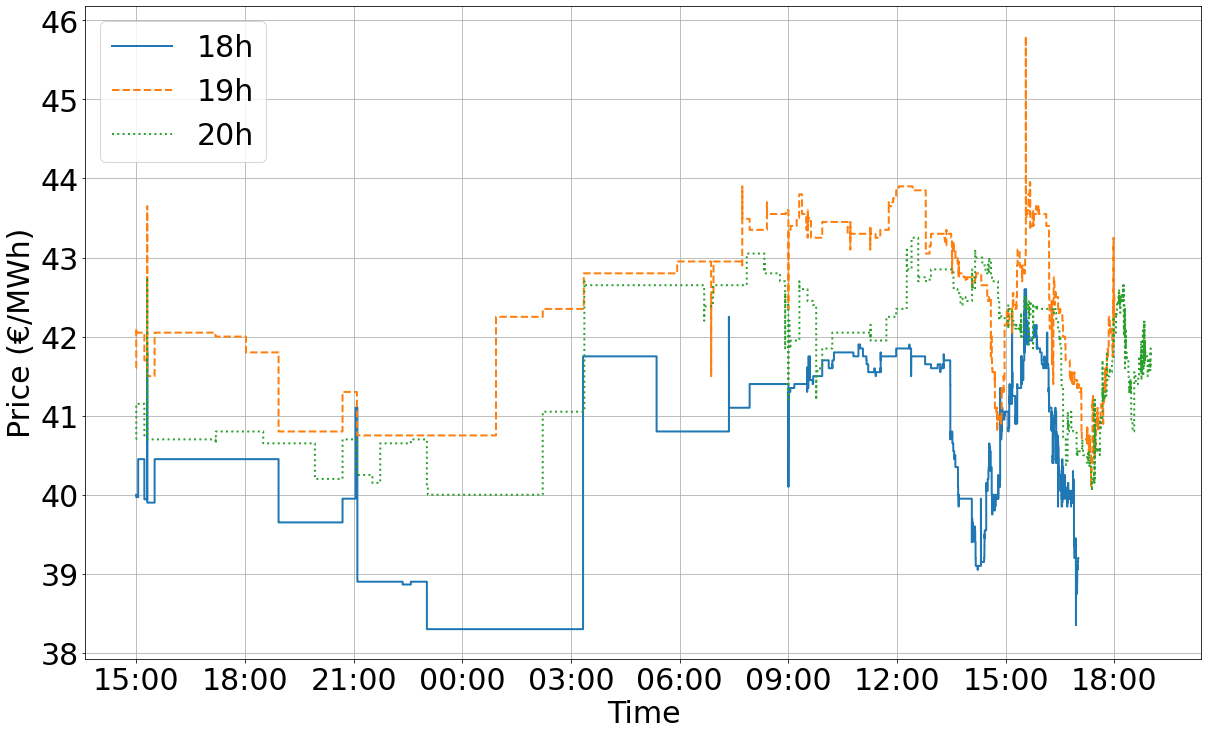}
    \caption{Prices on July 11\textsuperscript{th}, 2017}
\end{subfigure}
   \begin{subfigure}{0.9\textwidth}
    \centering
   \includegraphics[width=\textwidth]{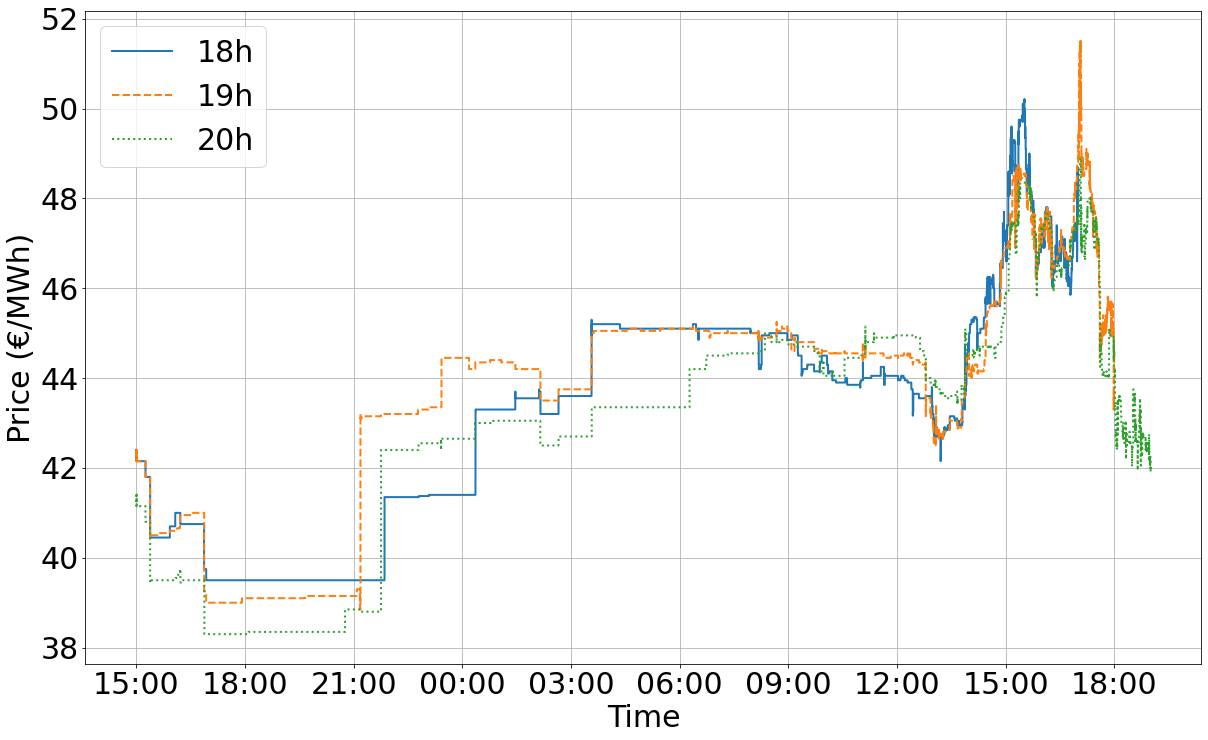}
    \caption{Prices on August 30\textsuperscript{th}, 2017}
    \end{subfigure}
       \caption{\label{fig:prices_ts} Intraday mid-prices for deliveries at 18h, 19h and 20h up to 1 hour before maturity for two trading sessions}
\end{figure}

\subsection{Arrival times of price changes distribution} 
\label{sec:intensity}

In this section, we study statistical properties of the price changing times, that is when there is an upward or downward move in the price. Intensity is first studied, followed by the (non) adequacy of the arrival times to an inhomogeneous Poisson process.

\medskip
\paragraph{\it Increasing intensity} Looking at Figure~\ref{fig:prices_ts}, the first empirical finding is that market activity increases when trading time approaches maturity. Let $N^c_t$ be the number of price changes between $0$ and $t$. Figure~\ref{fig:cum_intensity} represents the path of $(N^c_t)_t$ for one trading session but also the average for all trading sessions. $N^c_t$ is an estimator of the cumulative intensity up to time $t$. If the intensity were constant, $(N^c_t)_t$ would be a linear function of time. Let $(\lambda^c_t)_t$ be the intensity associated to the jump process $(N^c_t)_t$. We pointly estimate $(\lambda^c_t)_t$ using the kernel estimator
\[
\hat{\lambda}^c_t = \frac{\int_{0}^T \mathcal{K}_h(t-s)dN^c_s}{\int_{0}^T \mathcal{K}_h(t-s)ds} 
\]
with $T$ the last date of observation (that is one hour before maturity), $\mathcal{K}_h(u) = \frac{1}{h} \mathcal{K}(\frac{u}{h})$ for $u \in \mathbb{R}$, $h > 0$ and $\mathcal{K}$ a kernel function. The denominator allows to avoid boundary effects when the number of data diminishes. Figure~\ref{fig:kernel_intensity} shows the intensity process for a single trading session and the average of all the intensity processes using the Epanechnikov kernel $\mathcal{K}(u) = \frac{3}{4}(1-u^2){\bf 1}_{|u|\leq1}$ and the bandwidth $h = 300$ seconds. Typically, market activity is almost null during the first hours of the trading session, and at some point, intensity begins increasing exponentially as trading time comes closer to maturity. The quality of renewable production and consumption forecasts improves only a few hours before the delivery, so that market actors do not really have enough information to reduce their imbalances during the first hours of the trading session. Those empirical findings are consistent with the ones of \cite{Favetto2019, Graf2020} showing an increasing market activity in the order book with trading time.

\begin{figure}
    \centering
     \centering
    \begin{subfigure}{0.49\textwidth}
    \centering
    \includegraphics[width=\textwidth]{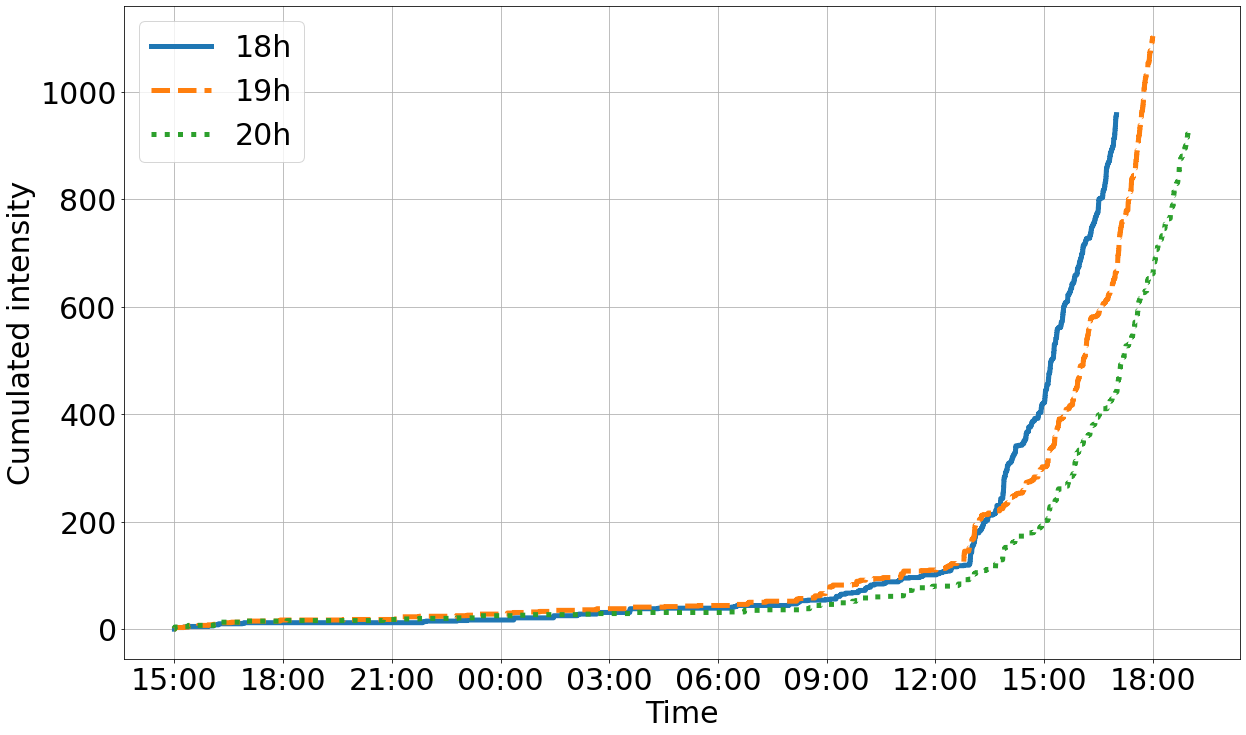}
    \caption{August 30\textsuperscript{th}, 2017}
\end{subfigure}
   \begin{subfigure}{0.49\textwidth}
    \centering
   \includegraphics[width=\textwidth]{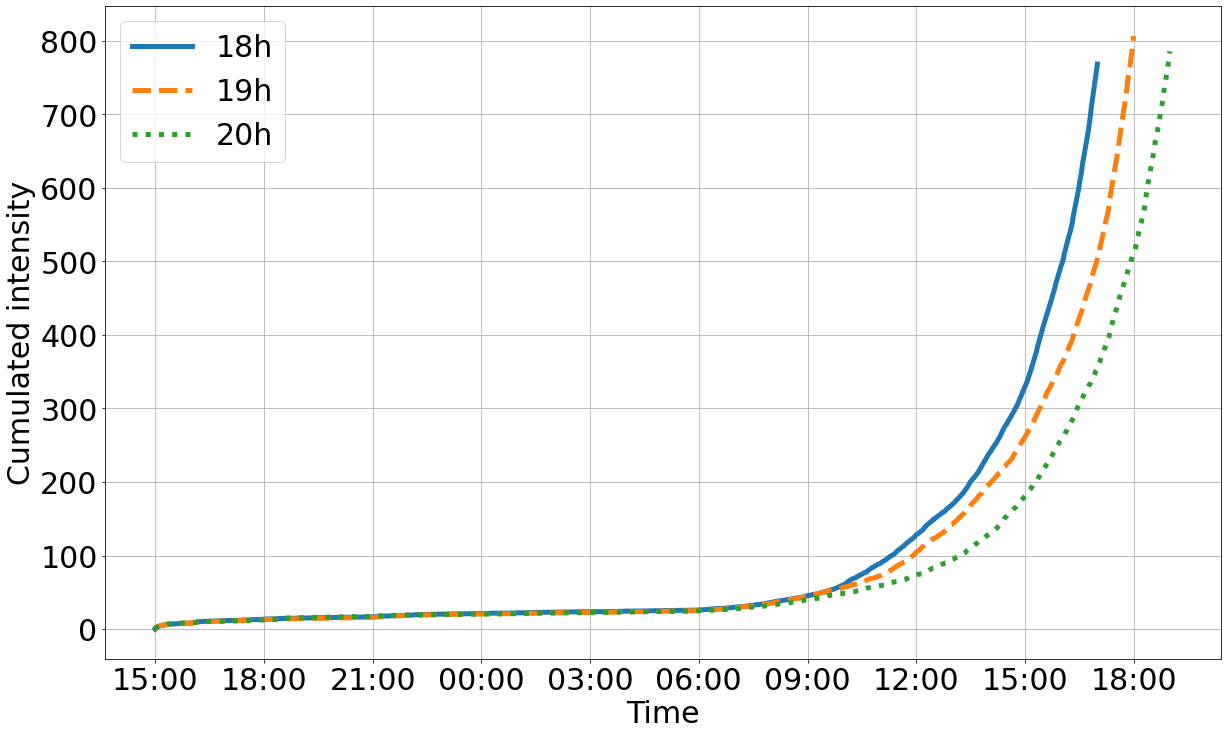}
    \caption{Average}
    \end{subfigure}
       \caption{\label{fig:cum_intensity} Number of price changes for deliveries at 18h, 19h and 20h up to 1 hour before maturity for one trading session and on average over all the dataset}
\end{figure}

\begin{figure}
    \centering
    \begin{subfigure}{0.49\textwidth}
    \centering
    \includegraphics[width=\textwidth]{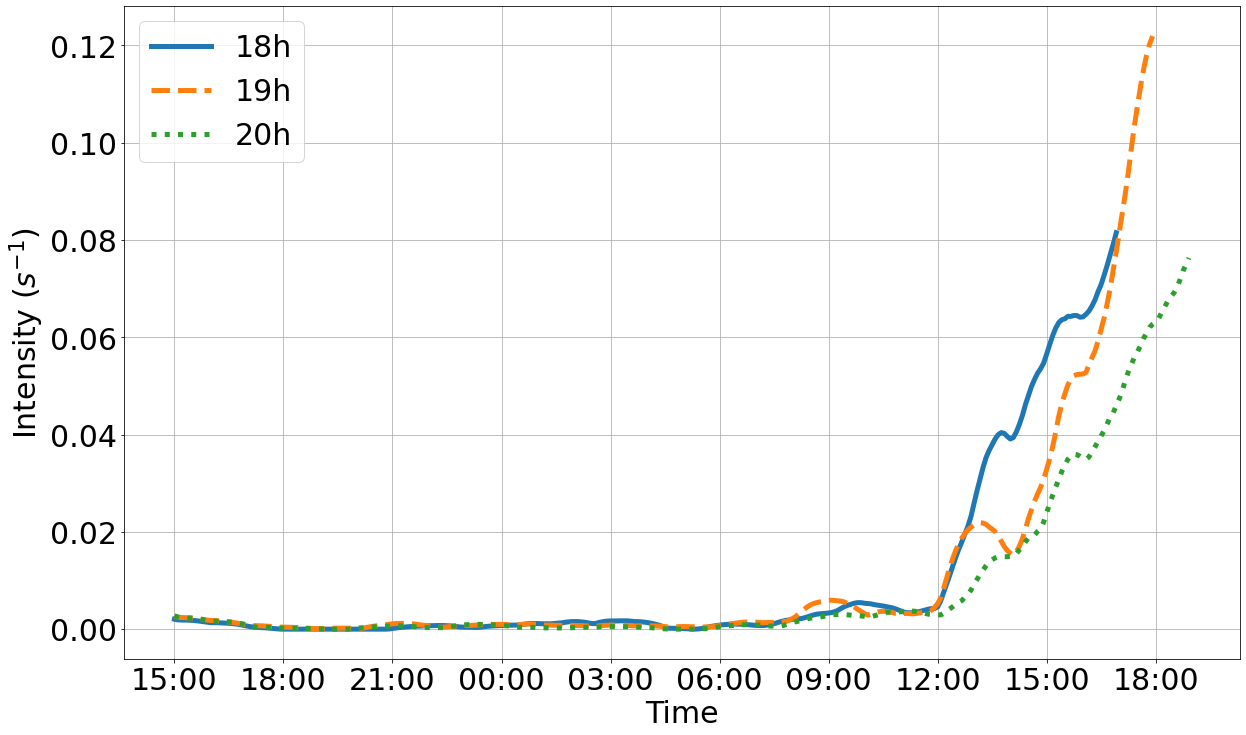}
    \caption{Intensity on August, 30\textsuperscript{th}, 2017}
\end{subfigure}
   \begin{subfigure}{0.49\textwidth}
    \centering
   \includegraphics[width=\textwidth]{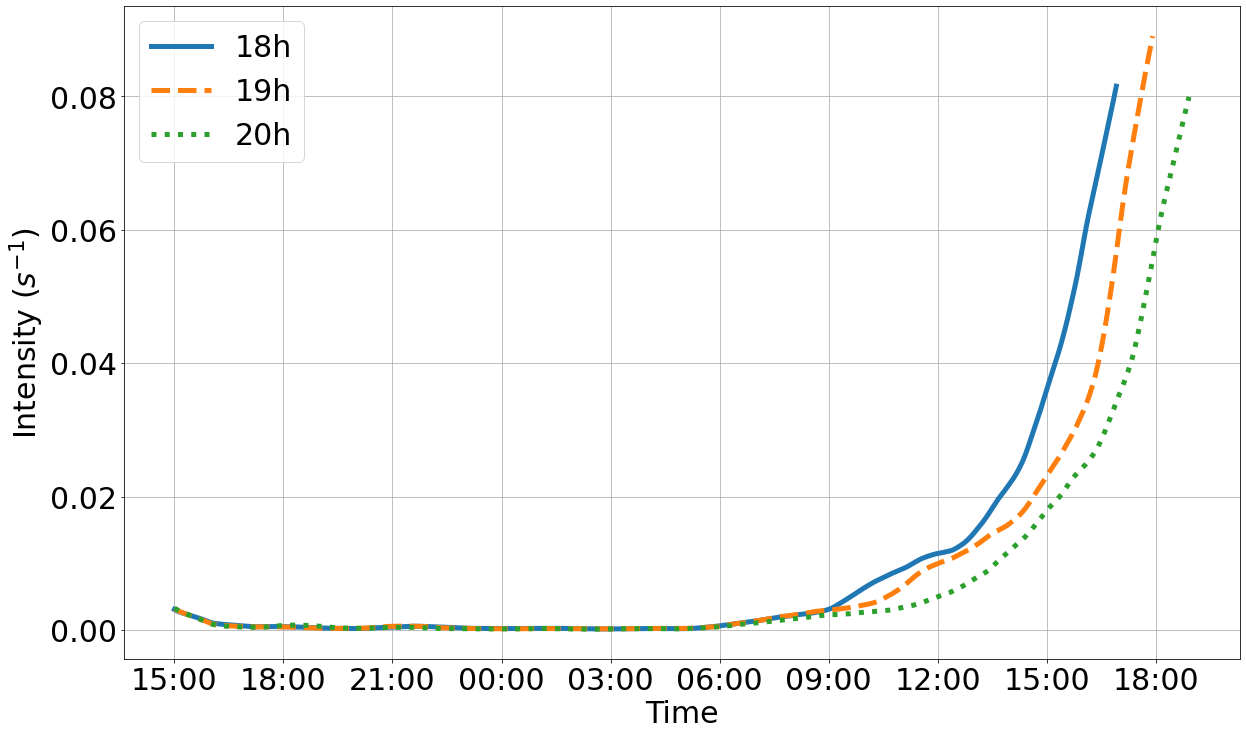}
    \caption{Average}
    \end{subfigure}
       \caption{\label{fig:kernel_intensity} Intensity of price changes for deliveries at 18h, 19h and 20h up to 1 hour before maturity for one trading session and on average over all the dataset}
\end{figure}

\medskip
\paragraph{\it Non Poissonian arrival times} Let us assume that $(N^c_t)_t$ is an inhomogeneous Poisson process with deterministic intensity $(\lambda^c_t)_t$. Thus, the variables $\Lambda^c(\tau^c_i) -\Lambda^c(\tau^c_{i-1})$, where $\Lambda^c(t) = \int_0^t \lambda^c_u du$ and $(\tau^c_i)_{i \geq 0}$ are the jump times of $N^c$, are identically and independently distributed (i.i.d.) and follow an exponential distribution with parameter 1. It is then possible to compute the quantile-quantile plot of those variables against an exponential distribution, see Figure~\ref{fig:qqplot18} for the quantile-quantile plot corresponding to one trading session and maturity 18h. Results are similar for maturities 19h and 20h. The function $\Lambda^c$ at time $t$ is estimated considering the empirical average of $N^c_t$ across the different days. The process is clearly not an inhomogeneous Poisson process, which argues in favor of a more sophisticated modeling. 

\begin{figure}
    \centering
    \includegraphics[width=0.6\textwidth]{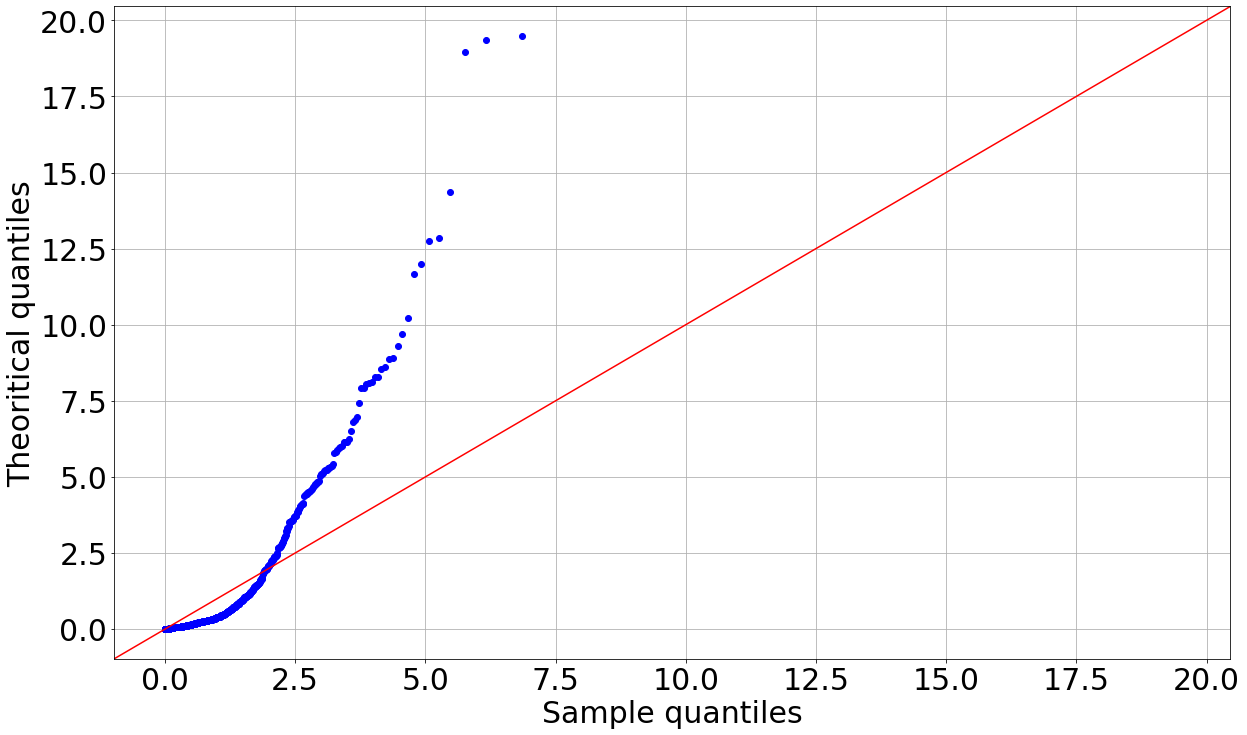}
    \caption{Quantile-quantile plot between the time-changed jump time intervals and an exponential distribution for the trading session of August, 30\textsuperscript{th}, 2017 and for maturity 18h in order to test if the process of price changes is an inhomogenous Poisson process}
    \label{fig:qqplot18}
\end{figure}

\subsection{Jump sizes distribution}
\label{sec:jumpsize}

While very liquid markets present tick by tick data, that is prices moving by one tick up or down, the electricity intraday market features jumps with different sizes. The purpose of this section is to study the distribution of the jump heights. We distinguish positive and negative jumps. When we study the negative jumps, we look at their absolute value. The distributions of positive and negative jumps are represented in Figure~\ref{fig:histjump} (left column) for maturities 18h, 19h and 20h.

\medskip
First, we can observe jump sizes up to 100~\euro/MWh. Those jumps are probably triggered by the lack of liquidity at the beginning of the trading session. As noticed by Balardy~\cite{Balardy2017}, the bid-ask spread is large when time is far from maturity. One order can then move the price strongly. If we consider only price moves less than 9 hours before maturity, the jump sizes are lower as seen in the histogram in Figure~\ref{fig:histjump} (right column). Figure~\ref{fig:statsttm} represents the mean and standard deviation of jump sizes when we stick to the jumps happening less than $x$ hours before the delivery starts, $x$ being the value on the x-axis. The mean and the standard deviation decrease when time approaches maturity, which is consistent with the increasing liquidity and decreasing bid-ask spread. There is a stabilization of those two moments between 8 and 10 hours before maturity. From now on, and in the rest of the paper, particular attention is paid to the price distribution starting 9~hours before maturity.

\medskip 
Second, we can observe that positive and negative jumps seem to have a similar distribution in Figure \ref{fig:histjump}. The Kolmogorov-Smirnov test to compare the two distributions, with null hypothesis being that the two sequences of samples are drawn for the same distribution \cite{Hodges1958}, is not rejected at level 95\% for maturity 18h, but it is for maturities 19h and 20h, considering the jump sizes from 9 hour before maturity. Performing this test on each day, it is not rejected for maturity 18h in $64.1~\%$ of the days in the dataset, $67.4~\%$ of them for maturity 19h and $77.2~\%$ of them for maturity 20h, at level $95\%$. 
Results are similar considering the whole trading session. 
Table~\ref{tab:secondmomentCI} gives confidence intervals for first and second orders moments at level $95\%$ for positive and negative jumps happening less than 9 hours before maturity. Intervals linked to positive and negative jumps intersect for both moments. The same table for the whole trading session could be displayed but it gives very wide intervals for the second order moments caused by the very large jumps. While we cannot perfectly conclude on the equality in distribution of positive and negative jumps, it seems to be a reasonable assumption if modeling the trading session starting 9 hours before maturity, especially as properties linked to second order moments will be of particular interest in our model.

\begin{figure}
    \centering
     \centering
    \begin{subfigure}{0.49\textwidth}
    \centering
    \includegraphics[width =\textwidth]{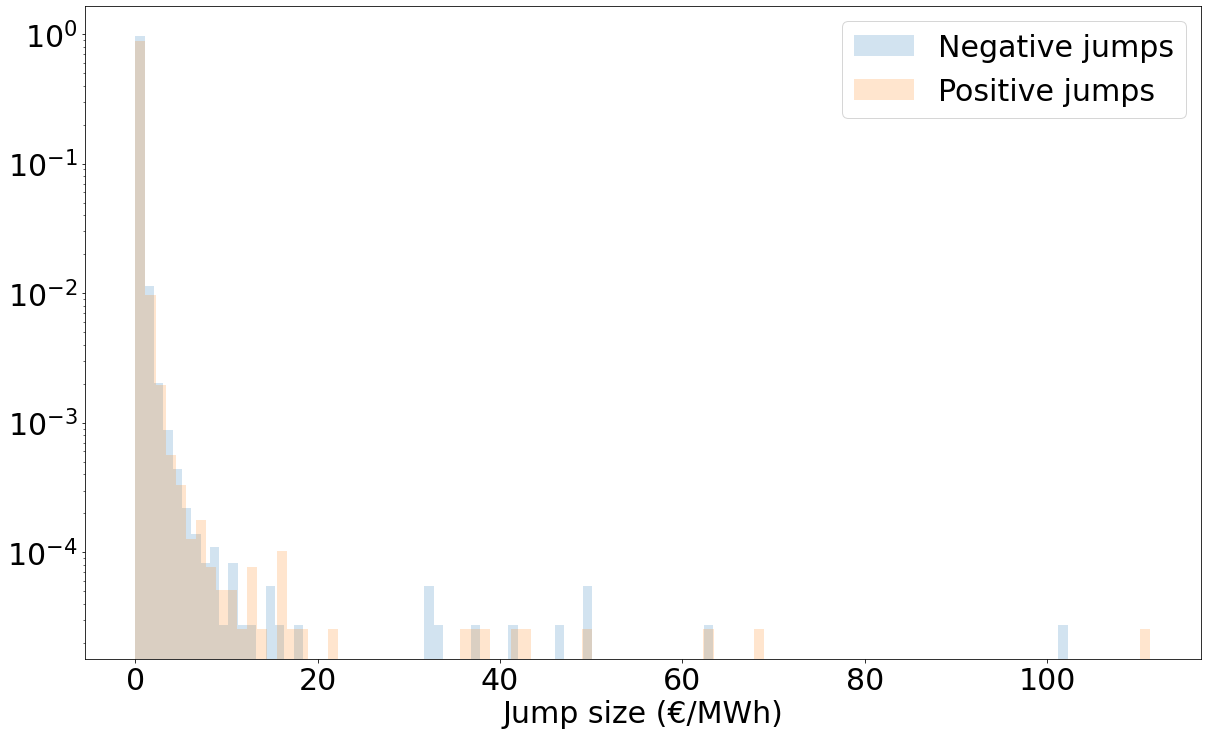}
    \caption{\label{fig:histjump18ttm0}  All trading sessions, maturity 18h}
\end{subfigure}
   \begin{subfigure}{0.49\textwidth}
    \centering
   \includegraphics[width=\textwidth]{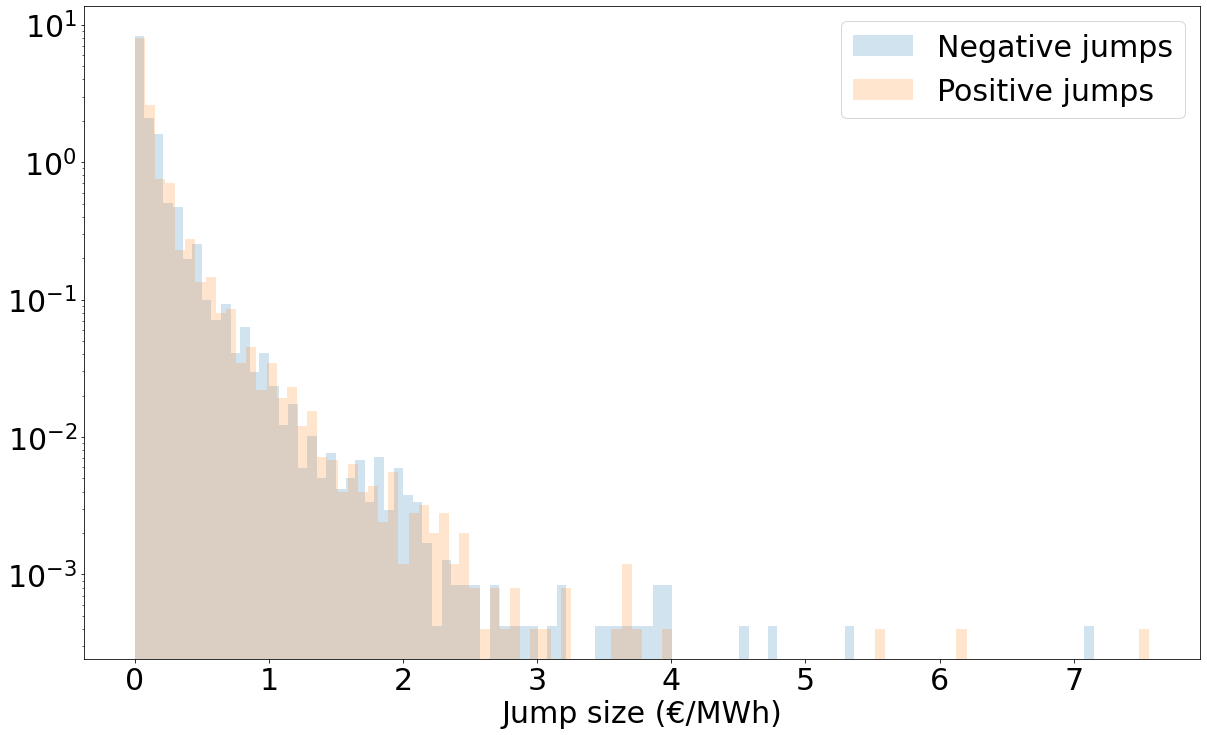}
    \caption{\label{fig:histjump18ttm9} 9 hours before maturity, maturity 18h}
    \end{subfigure}
    \begin{subfigure}{0.49\textwidth}
    \centering
    \includegraphics[width =\textwidth]{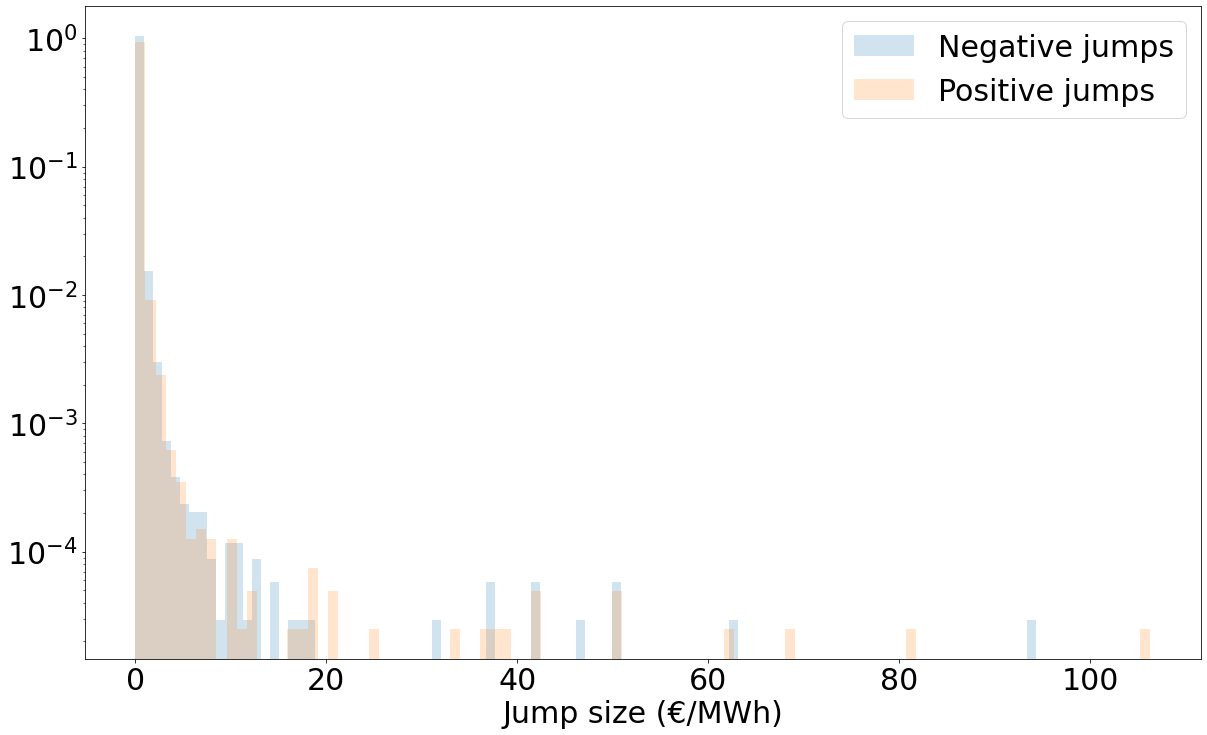}
    \caption{All trading sessions, maturity 19h}
\end{subfigure}
   \begin{subfigure}{0.49\textwidth}
    \centering
   \includegraphics[width=\textwidth]{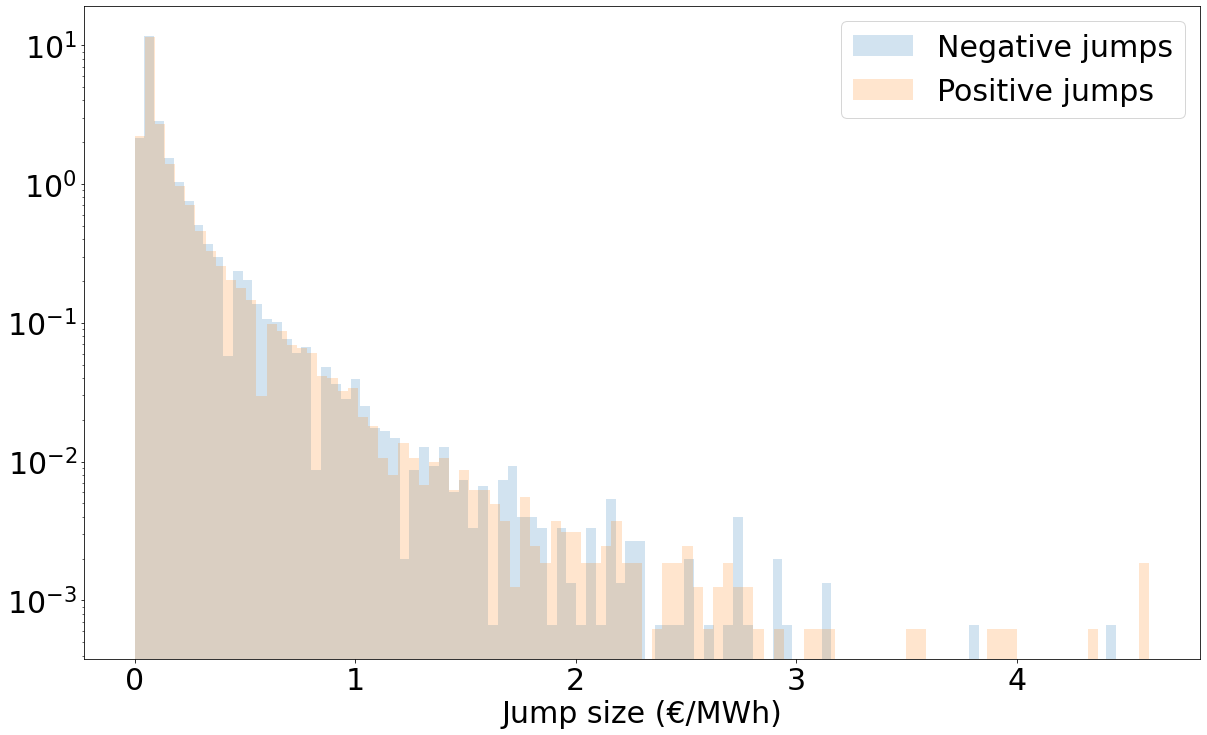}
    \caption{9 hours before maturity, maturity 19h}
    \end{subfigure}
    \begin{subfigure}{0.49\textwidth}
    \centering
    \includegraphics[width =\textwidth]{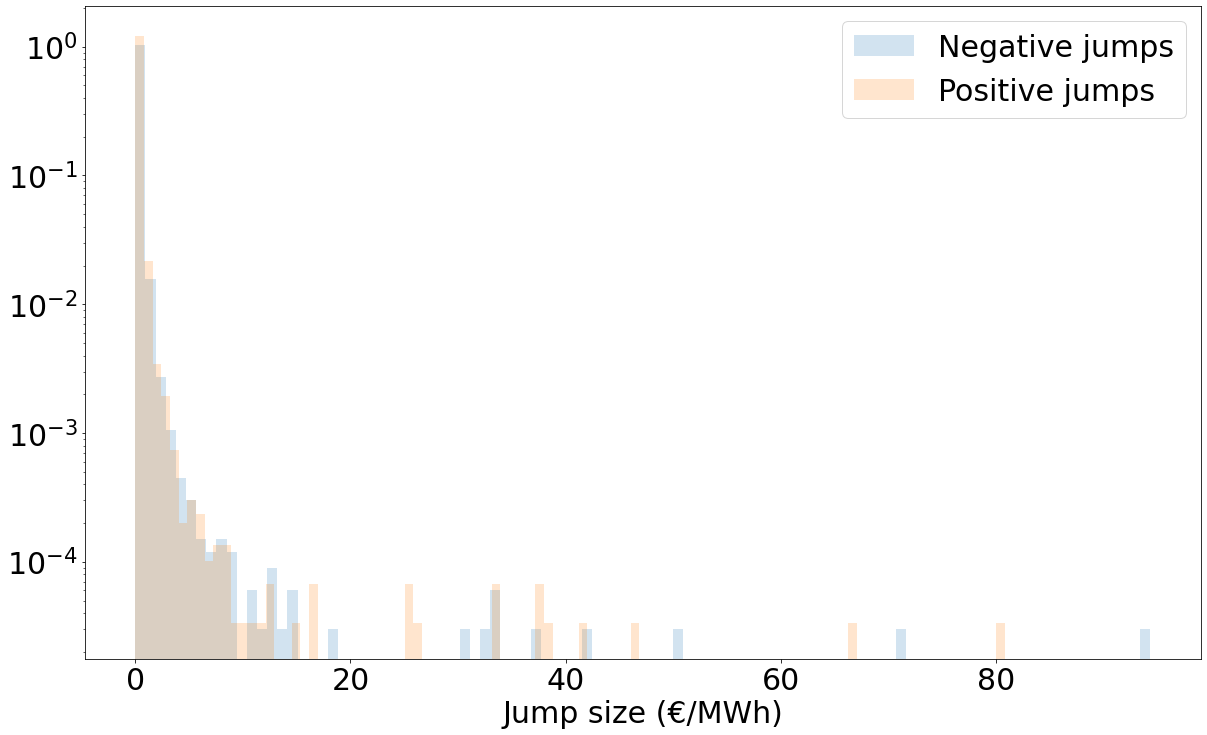}
    \caption{All trading sessions, maturity 20h}
\end{subfigure}
   \begin{subfigure}{0.49\textwidth}
    \centering
   \includegraphics[width=\textwidth]{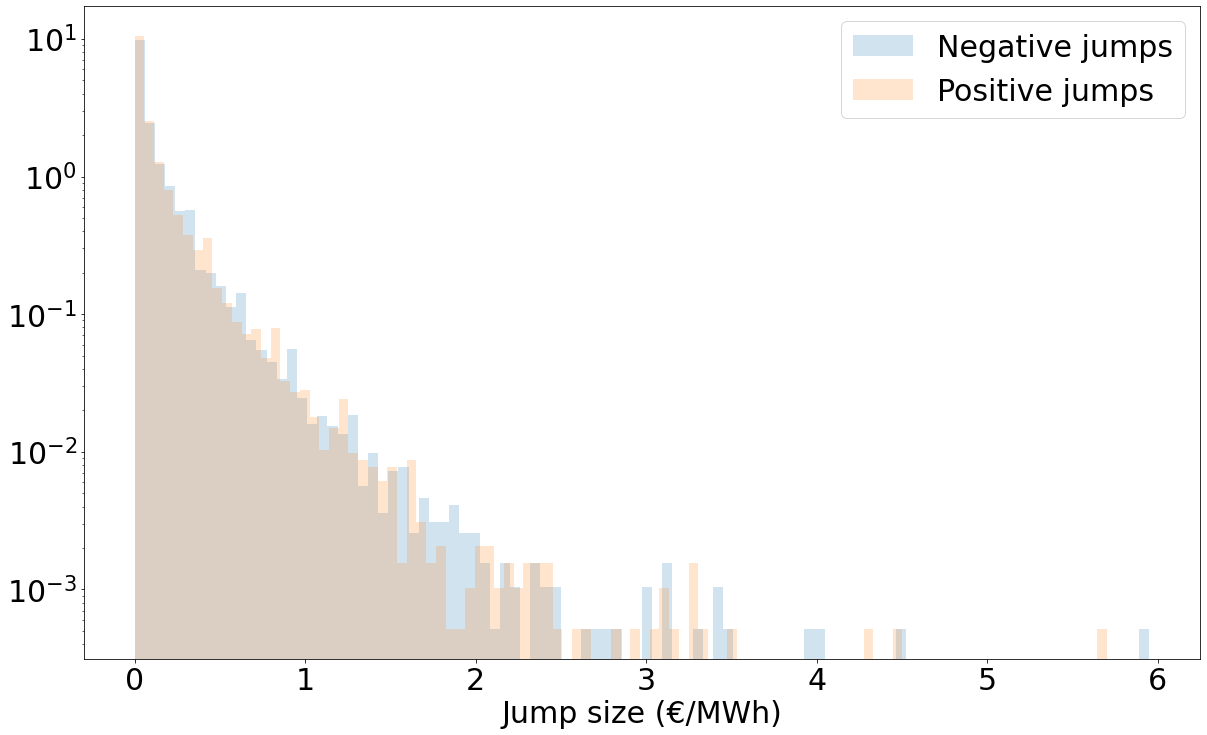}
    \caption{9 hours before maturity, maturity 20h}
    \end{subfigure}
       \caption{ \label{fig:histjump}Positive and negative jump size distributions with a log scale on the y-axis, for maturities 18h, 19h and 20h for all the trading sessions (left), keeping only jumps that happen less than 9 hours before maturity (right)}
\end{figure}

\begin{figure}
    \centering
     \centering
    \begin{subfigure}{0.49\textwidth}
    \centering
    \includegraphics[width =\textwidth]{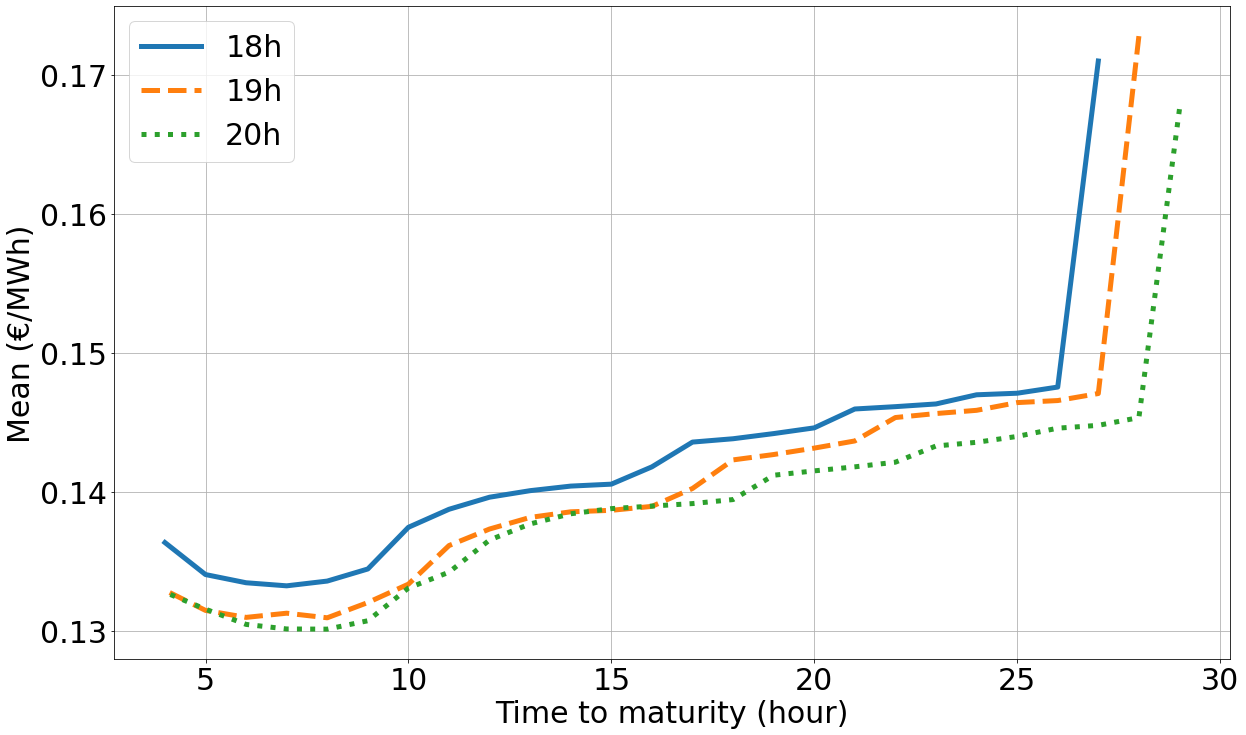}
    \caption{Mean}
\end{subfigure}
   \begin{subfigure}{0.49\textwidth}
    \centering
   \includegraphics[width=\textwidth]{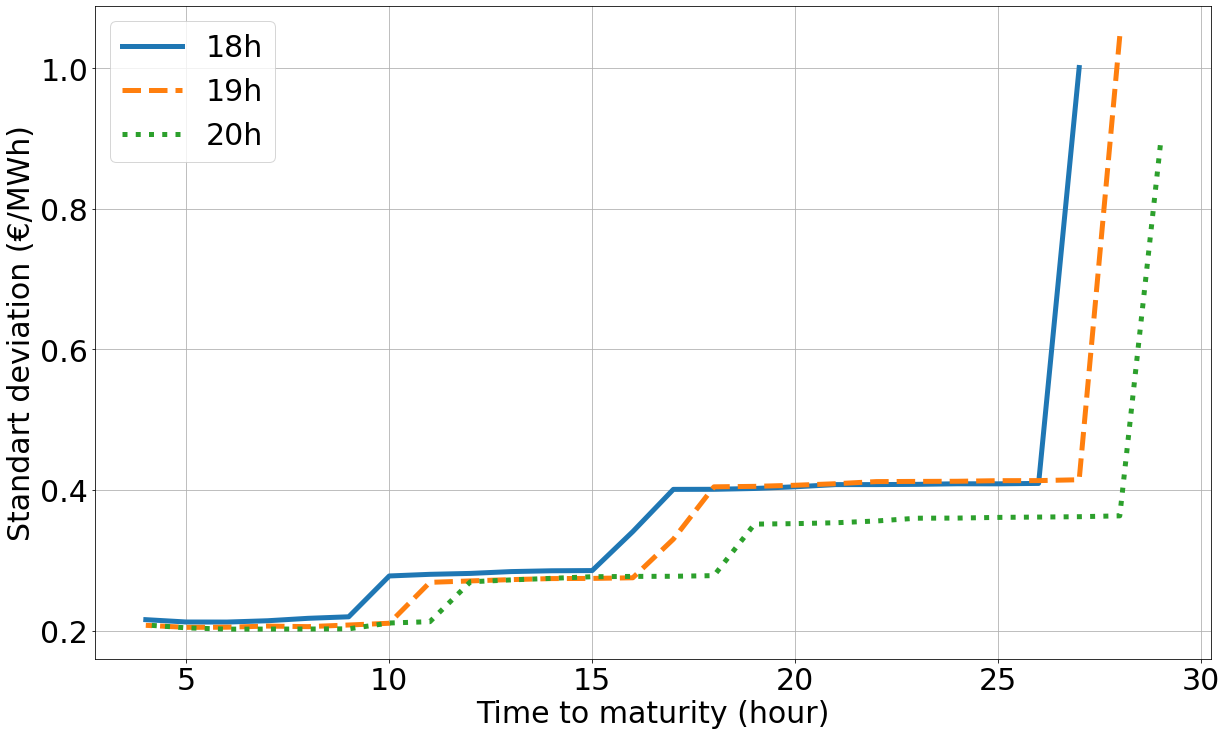}
    \caption{Standard deviation}
    \end{subfigure}
       \caption{Mean and standard deviation of jump sizes (positive and negative considered indifferently) against time to maturity: x-axis corresponds to the number of hours before maturity at which the estimation starts}
       \label{fig:statsttm}
\end{figure}

\begin{table}
    \centering
\begin{tabular}{lcccc}
\toprule
    \multirow{2}{*}{Hour / Stat.} &
      \multicolumn{2}{c}{Mean} &
      \multicolumn{2}{c}{Second order moment} \\
      & {Neg.} & {Pos.} & {Neg.} & {Pos.}  \\
      \midrule
          18 &[0.133,0.138] &  [0.131,0.136] &     [0.0613,0.0731] &      [0.0597,0.0711] \\
          19 &  [0.129,0.133] &  [0.131,0.136] &     [0.0574,0.0661] &     [0.0560,0.0631] \\
          20 &  [0.127,0.131] &   [0.130,0.135] &     [0.0527,0.0607] &     [0.0552,0.0638] \\
\bottomrule
\end{tabular}
    \caption{Confidence intervals for first and second order moments of negative and positive price jumps at level $95\%$ for the different maturities and considering jumps happening less than 9 hours before maturity}
    \label{tab:secondmomentCI}
\end{table}

\subsection{Signature plot}
\label{sec:signatureplot}
One of the main interesting facts concerning high-frequency data studied those last years is the presence of microstructure effects and in particular the existence of some patterns in the signature plot: the estimated realized volatility
\[\hat{C}(T, \delta) = \frac{1}{T} \sum_{i=1}^{\lfloor \frac{T}{\delta} \rfloor} (f_{i\delta} - f_{(i-1)\delta})^2, \; \delta > 0,\]
decreases with the estimation frequency, where $(f_t)_t$ is the price process observed between $0$ and $T$, and $\delta$ represents the estimation time step (inverse of the frequency). While it is natural to consider the largest frequency to estimate quadratic variation with a high precision, the estimator can increase a lot and be unstable for large frequencies. The signature plot is the function
\[C(T, \delta) = \mathbb{E}\left(\hat{C}(T, \delta) \right), \; \delta > 0.\] 
This phenomenon is called microstructure noise and it is caused by the fact that at high frequencies, prices have a mean reverting behavior ; a positive jump is often followed by a negative one and this round trip is not observed at lower frequencies. The objective is then to have a consistent model that can explain this volatility change: this is not the case if we consider a Brownian motion which is invariant by scale change and has a flat signature plot. The signature plot for the electricity intraday prices is represented in Figure~\ref{fig:signatureplot} considering the price process less than 9 hours before maturity for one trading session (left) and on average (right). This shape is common in the high-frequency finance literature, with a fast decrease when frequency goes high and a stabilization as it becomes low. Yet, to our knowledge, we are the first to exhibit such behavior for electricity intraday prices. The signature plot still has this shape when considering all the trading sessions but it is less smooth because of the big jumps happening at the beginning of the trading session.

\begin{figure}
    \centering
     \centering
    \begin{subfigure}{0.49\textwidth}
    \centering
    \includegraphics[width=\textwidth]{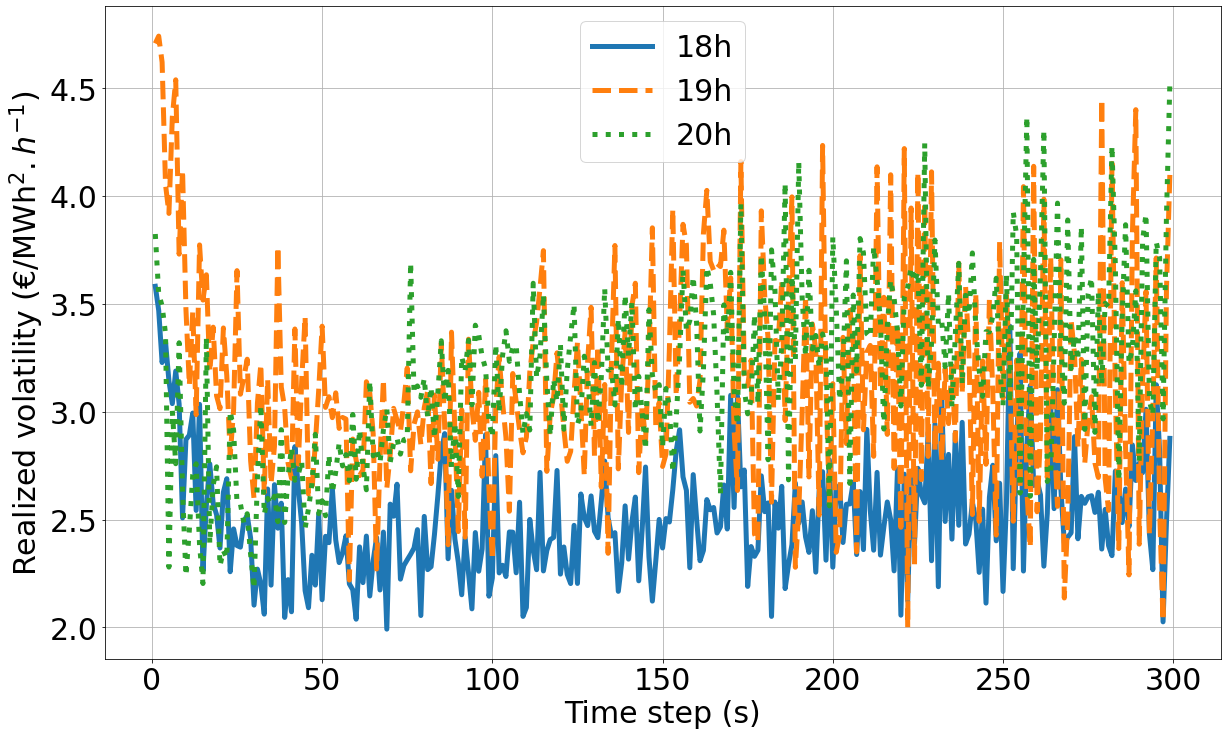}
    \caption{August 30\textsuperscript{th}, 2017}
\end{subfigure}
   \begin{subfigure}{0.49\textwidth}
    \centering
   \includegraphics[width=\textwidth]{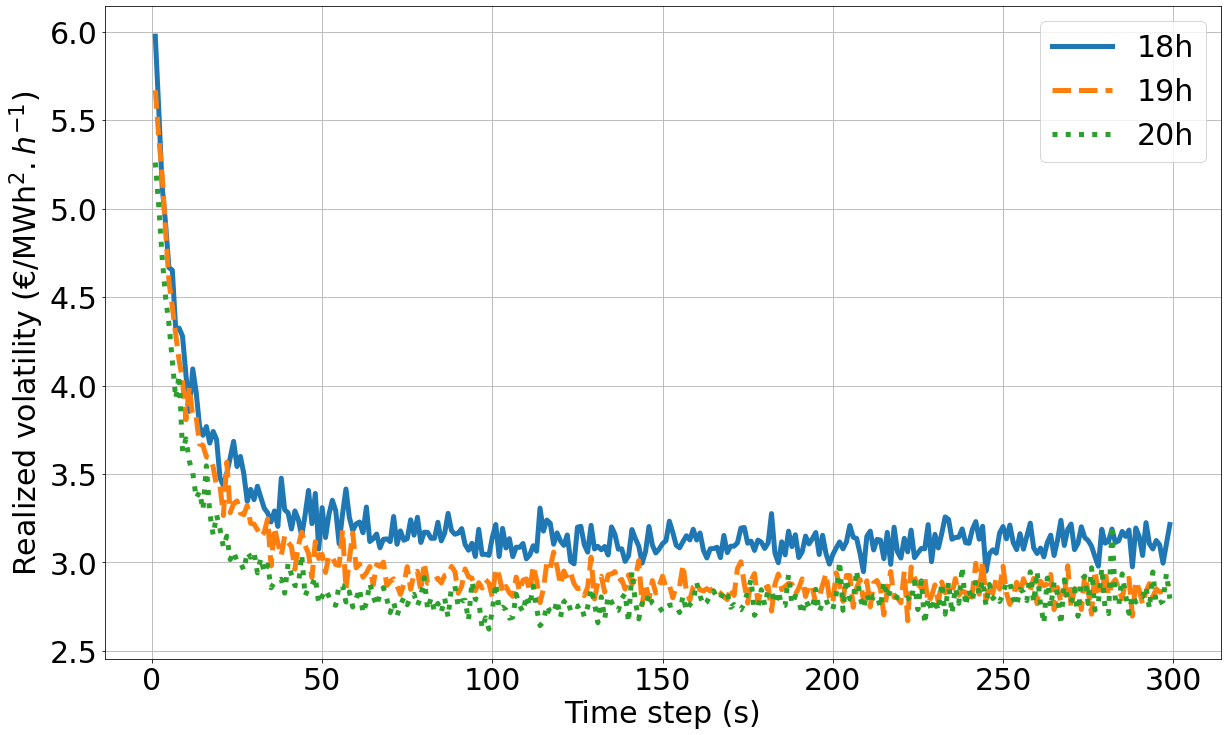}
    \caption{Average}
    \end{subfigure}
       \caption{\label{fig:signatureplot} Signature plot for different maturities estimated from 9 hours before maturity for one trading session, and on average over all trading sessions}
\end{figure}

\section{Price model}
\label{sec:model}

Let us consider the sequence of arrival times $0 < \tau_1 < \tau_2 < \ldots$ defined on a rich enough probability space $(\Omega, \mathcal{F}, \mathbb{P})$ endowed with a right continuous and complete filtration $(\mathcal{F}_t)_{t \geq 0}$. Let $(J_i)_{i \geq 1}$ be a sequence of positive i.i.d. random variables such that $J_{i}$ is $\mathcal{F}_{\tau_i}$ measurable for $i \geq 1$. We assume that they have the same law as a random variable $J$ defined on $(\Omega, \mathcal{F}, \mathbb{P})$ with $\mathbb{E}(J^2) < \infty$. Marking the arrival times $(\tau_i)_i$, one can construct two sequences $(\tau^+_i)_i$ and $(\tau^-_i)_i$, associated with two sequences of jump sizes $(J^+_i)_i$ and $(J^-_i)_i$, and define the bivariate point process $N = (N^+, N^-)^{\top}$ with
\[N^+_t = \sum_{n=1}^{\infty} {\bf 1}_{\tau^+_n \leq t},\]
\[N^-_t = \sum_{n=1}^{\infty} {\bf 1}_{\tau^-_n \leq t}.\]
From now on, we work on the finite time horizon $\left[0,T\right]$, $T > 0$. We model the point process $N$ as a bivariate Hawkes process whith intensity depending on the marks $(J_i)_{i}$:
\begin{equation} \label{eq:intensities}
\begin{pmatrix}\lambda_t^+\\ \lambda_t^-
\end{pmatrix} =  \mu\left(\frac{t}{T}\right) \begin{pmatrix} 1 \\ 1 \end{pmatrix} + \int_0^t \varphi(t-s)\begin{pmatrix}J_sdN_s^+\\ J_sdN_s^-\end{pmatrix}\end{equation}
with $\mu : \left[0,1\right] \to \mathbb{R}_+$ a non decreasing bounded function and $\varphi : \mathbb{R}_+ \to \mathbb{R}^{2,2}$ a locally bounded function with positive components such that $\rho(K) < 1$, where $K = \mathbb{E}(J)\int_0^{\infty}|\varphi(u)|du$ and $\rho(K)$ is the spectral radius of $K$. Recall that the standard notation $\int_0^t h(s)J_sdN^+_s$ (resp. $\int_0^t h(s)J_sdN^-_s$) for a measurable function $h$ stands for $\sum_{i=1}^{N_t^+} J_i^+ h(\tau^+_i)$ (resp. $\sum_{i=1}^{N_t^+} J_i^- h(\tau^-_i)$). Let us also consider the marked version of $N$, $(f^+, f^-)^{\top}$, that models the upward and downward jumps in the price by 
\begin{equation} \label{eq:upwarddownward}
\begin{pmatrix} f^+_t \\ f^-_t \end{pmatrix} = \int_0^t \begin{pmatrix} J_s dN^+_s \\ J_s dN^-_s \end{pmatrix}.
\end{equation}
The price is then given by
\begin{equation} \label{eq:prices}
f_t = f_0 + f_t^+ - f_t^-
\end{equation}
with $f_0 \in \mathbb{R}$ the initial value of the price. 

\begin{remark} The marked Hawkes process could also be defined using random measure notations, see Jacod and Shiryaev~\cite[Chapter 2]{Jacod13}. Let $N(dt,dx) = \left(N^+(dt,dx), N^-(dt,dx)\right)^{\top}$ be the 2-dimensional random Poisson measure with compensator $\nu(dt, dx) = \lambda_t dt \otimes \mu_J(dx)$ on $\left[0,T\right] \times \mathbb{R}_+$ where $\mu_J(dx)$ is a probability measure representing the law of $J$ and 
\[\lambda_t =  \mu\left(\frac{t}{T}\right) \begin{pmatrix} 1 \\ 1 \end{pmatrix} + \int_0^t \int_{\mathbb{R}_+} x \varphi(t-s) N(ds,dx).\]
The components of the price process are then defined by 
\[\begin{pmatrix} f^+_t \\ f^-_t \end{pmatrix} = \int_0^t \int_{\mathbb{R}_+} x N(ds,dx)\]
and 
\[f_t = f_0 + \int_0^t \int_{\mathbb{R}_+} x \left(N^+(ds, dx) - N^-(ds, dx)\right).\]
One can then naturally define the integral, for a measurable function $h$, 
\[\int_0^t \int_{\mathbb{R}_+} h(s,x) N(ds,dx) = \begin{pmatrix} \sum_{i=1}^{N_t^+} h(\tau_i^+, J_i^+) \\ \sum_{i=1}^{N_t^-} h(\tau_i^-, J_i^-) \end{pmatrix}.\]
While this notation can seem more convenient than the notation $J_sdN_s$, we prefer keeping the latter as it is more used in financial price modeling or insurance risk modeling.
\end{remark}

\medskip
A priori, our model allows to represent the different empirical facts identified on data in Section~\ref{sec:stylizedfact}:
\begin{itemize}
    \item[(i)] the baseline intensity depends on time, in order to model the increasing market activity over a trading session identified in Section~\ref{sec:intensity};
    \item[(ii)] the Hawkes modeling encompasses the inhomogeneous Poisson framework which does not fit the data well, see Section~\ref{sec:intensity};
    \item[(iii)] positive and negative jump heights in price are random with the same distribution, see Section~\ref{sec:jumpsize};
    \item[(iv)] the mutual excitation between positive and negative jumps is expected to represent the signature plot described in Section~\ref{sec:signatureplot} well, as in Bacry et al.~\cite{Bacry2013a}.
\end{itemize}

\begin{remark}
The time in the baseline function is normalized by the time horizon $T$. The normalization allows to have a bounded intensity even when the time horizon becomes large. It is then possible to study the limit behavior of the model when $T \to \infty$ (macroscopic scale), which is done in Sections~\ref{sec:theoriticalprop} and \ref{sec:macro}. This normalization is used by Duval and Hoffmann~\cite{Duval11} to study the limit behavior of an inhomogeneous compound Poisson process when $T \to \infty$.
\end{remark}

We choose to set the model as in Assumption~\ref{ass:parametric} in order to have analytical and tractable formulas while still allowing to represent the different empirical facts. Our model presents two main differences with the classic Hawkes model used by Bacry et al.~\cite{Bacry2013a}, which are the presence of random jumps, both in prices and intensities, and the time dependent baseline intensity. This parameterization is the same as the one in Graf von Luckner and Kiesel~\cite{Graf2020} to model the order book activity. The mid-price considered here being strongly related to the order book activity, the empirical study in \cite{Graf2020} is in favor of the choice of the parameterization~\ref{ass:parametric}. The main difference made by our model is the presence of the jump size in the intensity and the null diagonal in $\varphi$ : this last assumption allows us to have a tractable model with closed formulas for the moments and the signature plot, as shown in Section~\ref{sec:theoriticalprop}. 

\begin{assumption}
\label{ass:parametric}
The baseline intensity and the excitation function are given by
\begin{itemize}
\item[(i)] $\mu(t) = \mu_0 e^{\kappa t}$ with $\mu_0$, $\kappa > 0$;
\item[(ii)] $\varphi = \begin{pmatrix} 0 & \varphi_{\exp}\\ \varphi_{\exp} & 0 \end{pmatrix}$ with $\varphi_{\exp}: t \mapsto \alpha e^{-\beta t}$, $\alpha$, $\beta > 0$ ; the condition $\rho(K) < 1$ becomes $\alpha \mathbb{E}(J) < \beta$.
\end{itemize}
\end{assumption}

\subsection{Theoretical properties}
\label{sec:theoriticalprop}
In this section, one confirms theoretically that our model is able to feature an increasing market activity, by computing $\mathbb{E}(\lambda^+_t + \lambda^-_t)$ using Proposition~\ref{prop:expectation}, and that it can provide a good signature plot representation using Proposition~\ref{prop:signatureplot}, under Assumption~\ref{ass:parametric}. First and second order moment properties for the price are also given in Proposition~\ref{prop:expectation} and Proposition~\ref{prop:moment2}.

\medskip
Proposition~\ref{prop:expectation} gives the expectation of the positive and negative price changes together with the expectation of the intensities. Proof is given in Section~\ref{proof:propexpectation}. The market activity intensity is equal to $\mathbb{E}(\lambda_t^+ + \lambda_t^-)$, which is increasing exponentially with time at a rate $\frac{\kappa}{T}$: our model can reproduce the increase in market activity intensity.

\begin{proposition} \label{prop:expectation} Let us consider the model \eqref{eq:intensities}--\eqref{eq:upwarddownward}--\eqref{eq:prices} under Assumption~\ref{ass:parametric}. We have for $t \in \left[0,T\right]$
\begin{align*}
\mathbb{E}(f^+_t) &= \mathbb{E}(f^-_t) =\mu_0\mathbb{E}(J)\left( \frac{\beta + \frac{\kappa}{T}}{\frac{\kappa}{T}\left(\beta - \alpha \mathbb{E}(J) + \frac{\kappa}{T}\right)}e^{\kappa \frac{t}{T}}\right.\\
&\hspace{3.25cm}\left.+ \frac{\alpha \mathbb{E}(J)}{\left(\beta-\alpha\mathbb{E}(J)\right)\left(\beta-\alpha\mathbb{E}(J) + \frac{\kappa}{T}\right)}e^{-(\beta-\alpha\mathbb{E}(J))t} - \frac{\beta}{\frac{\kappa}{T}(\beta - \alpha\mathbb{E}(J))}\right)
\end{align*}
and
\[
\mathbb{E}(\lambda_t^+) = \mathbb{E}(\lambda_t^-) = \mu_0 \left(\frac{\beta + \frac{\kappa}{T}}{\beta - \alpha \mathbb{E}(J) + \frac{\kappa}{T}}e^{\kappa \frac{t}{T}} - \frac{\alpha \mathbb{E}(J)}{\beta - \alpha \mathbb{E}(J) + \frac{\kappa}{T}} e^{-(\beta - \alpha  \mathbb{E}(J))t}\right).
\]
\end{proposition}

Proposition~\ref{prop:moment2} gives the second order moment of the price. Proof is given in Section~\ref{proof:propmoment2}.

\begin{proposition} \label{prop:moment2} Let us consider the model \eqref{eq:intensities}--\eqref{eq:upwarddownward}--\eqref{eq:prices} under Assumption \ref{ass:parametric}. We have for $t \in \left[0,T\right]$
\begin{align*}
\mathbb{E}(f_t^2) &= f_0^2 + 2\mu_0\mathbb{E}(J^2)\left(\left(C_{1} + C_{2} + C_{3} + C_{4}\right)e^{\kappa \frac{t}{T}} - C_{1} e^{-(\beta-\alpha \mathbb{E}(J))t} \right.\\
&\hspace{2.9cm}\left.- C_{2} e^{-2(\beta+\alpha \mathbb{E}(J))t} - C_{3}e^{-(\beta+\alpha \mathbb{E}(J))t} - C_{4}\right)
\end{align*}
with 
\begin{align*}C_{1} &= \frac{- \alpha^2 \mathbb{E}(J)^2}{\left(\beta - \alpha \mathbb{E}(J)\right) \left(\beta+3\alpha \mathbb{E}(J)\right)\left(\beta - \alpha \mathbb{E}(J) + \frac{\kappa}{T}\right)}, \\
C_{2} &= \frac{\alpha^2 \mathbb{E}(J)^2\left(\beta + 2\alpha\mathbb{E}(J)\right)}{\left(\beta + \alpha\mathbb{E}(J)\right)^2\left(\beta + 3\alpha\mathbb{E}(J)\right)\left(2\beta + 2\alpha \mathbb{E}(J) + \frac{\kappa}{T}\right)},\\
C_{3} &= \frac{\alpha \beta \mathbb{E}(J)}{\left(\beta+\alpha \mathbb{E}(J)\right)^2\left(\beta + \alpha \mathbb{E}(J) + \frac{\kappa}{T}\right)},\\
C_{4} &= \frac{\beta^3}{\frac{\kappa}{T}\left(\beta+\alpha\mathbb{E}(J)\right)^2\left(\beta-\alpha\mathbb{E}(J)\right)}.
\end{align*}
\end{proposition}

In Proposition~\ref{prop:signatureplot}, we give a tractable formula that allows to compute the signature plot. Proof is given in Section~\ref{proof:propsignatureplot}.

\begin{proposition} \label{prop:signatureplot} Let us consider the model \eqref{eq:intensities}--\eqref{eq:upwarddownward}--\eqref{eq:prices} under Assumption~\ref{ass:parametric} and $0 \leq s < t \leq T$. We have 
\[
\mathbb{E}\left(\left(f_t-f_s\right)^2\right) = \mathbb{E}(f_t^2) - \mathbb{E}(f_s^2) -\frac{(1 - e^{-(\beta+\alpha\mathbb{E}(J))(t-s)})}{\beta + \alpha\mathbb{E}(J)} \left(\frac{d \mathbb{E}(f_s^2)}{ds}(s) - 2\mathbb{E}(J^2)\mathbb{E}(\lambda_s^+)\right)
\]
with $\mathbb{E}(\lambda_s^+)$ given in Proposition~\ref{prop:expectation} and $\mathbb{E}(f_s^2)$ given in Proposition~\ref{prop:moment2}.
\end{proposition}

The signature plot for $t \in \left[0,T\right]$, $\delta > 0$, given by
\[C(t,\delta) = \frac{1}{t}\mathbb{E}\left(\sum_{i=1}^{\lfloor \frac{t}{\delta}\rfloor} \left(f_{i\delta} - f_{(i-1)\delta}\right)^2 \right),\]
can be computed directly from the result of Proposition~\ref{prop:signatureplot}:
\begin{equation}\label{eq:Cnotcomputed}C(t,\delta) = \frac{1}{t} \left(\mathbb{E}(f^2_{\lfloor \frac{t}{\delta}\rfloor \delta})-f_0^2\right) - \frac{1}{t}\frac{(1 - e^{-(\beta+\alpha\mathbb{E}(J))\delta})}{(\beta + \alpha\mathbb{E}(J))} \sum_{i=0}^{\lfloor \frac{t}{\delta}\rfloor -1 } \left(\frac{d \mathbb{E}(f_s^2)}{ds}(i\delta) - 2\mathbb{E}(J^2)\mathbb{E}(\lambda_{i\delta}^+)\right).
\end{equation}
$C(t,\delta)$ can then be computed explicitly directly from Equation~\eqref{eq:Cnotcomputed}. In particular, the following two regimes are specifically interesting to understand the evolution of the signature plot:
\begin{itemize}
    \item In the microscopic regime, that is $\delta \to 0$, the signature plot converges to 
    \begin{equation} \label{eq:Cmicro}
    C^{micro}(t) = 2\mathbb{E}(J^2) \frac{\mathbb{E}\left(\int_0^t \lambda^+_s ds\right)}{t}.
    \end{equation}
    \item In the macroscopic regime, that is $\delta \to \infty$ and $\frac{\delta}{t}\to 0$ (while $t\leq T$), we have 
    \begin{equation}\label{eq:macro}
    C^{macro}(t) \sim \frac{2 \mathbb{E}(J^2)}{\left(1 + \frac{\alpha \mathbb{E}(J)}{\beta}\right)^2 \left(1-\frac{\alpha \mathbb{E}(J)}{\beta}\right)} \frac{\int_0^t \mu(\frac{s}{T})ds}{t}.
    \end{equation}
\end{itemize}
The microscopic and macroscopic signature plots evolve at the same speed as the intensity with respect to time, that is exponentially with a rate $\frac{\kappa}{T}$. This volatility increase, present at both microscopic and macroscopic scales, is induced by the time-dependent baseline and allows to represent the so-called {\it Samuelson effect}, which is well known for electricity forward prices, see for instance Jaeck and Lautier~\cite{Jaeck2016} : volatility increases when time to maturity decreases. One can also give intuition about the formula considering the asymptotics $t \to \infty$ (while $t \leq T$) : 
\[C(t,\delta) \sim \frac{2 \mathbb{E}(J^2) \int_0^t \mu(\frac{s}{T})ds}{t\left(1-\frac{\alpha\mathbb{E}(J)}{\beta}\right)} \left(\frac{1}{\left(1 + \frac{\alpha \mathbb{E}(J)}{\beta}\right)^2} + \left(1 - \frac{1}{\left(1 + \frac{\alpha \mathbb{E}(J)}{\beta}\right)^2}\right)\left(\frac{1-e^{-(\beta + \alpha \mathbb{E}(J))\delta}}{\left(\beta + \alpha \mathbb{E}(J)\right)\delta}\right)\right).
\]
In this asymptotics, the signature plot is a decreasing function of $\delta$ with a term having an exponential decay and a constant term. This theoretical shape can fit the empirical signature plot described in Section~\ref{sec:signatureplot} for $t = T$ as we have
\[C(T, \delta) \sim \frac{2 \mathbb{E}(J^2) \int_0^1 \mu(s)ds}{1-\frac{\alpha\mathbb{E}(J)}{\beta}} \left(\frac{1}{\left(1 + \frac{\alpha \mathbb{E}(J)}{\beta}\right)^2} + \left(1 - \frac{1}{\left(1 + \frac{\alpha \mathbb{E}(J)}{\beta}\right)^2}\right)\left(\frac{1-e^{-(\beta + \alpha \mathbb{E}(J))\delta}}{\left(\beta + \alpha \mathbb{E}(J)\right)\delta}\right)\right).
\]
In the case where $\mu$ is constant and equal to $\mu_0$, and the regime is stationary, we find a signature plot equal to
\begin{equation} \label{eq:signaturestati}
C(T, \delta) = \frac{2\mu_0\mathbb{E}(J^2)}{1- \frac{\alpha \mathbb{E}(J)}{\beta}}\left(\frac{1}{\left(1 + \frac{\alpha \mathbb{E}(J)}{\beta}\right)^2} + \left(1 - \frac{1}{\left(1 + \frac{\alpha \mathbb{E}(J)}{\beta}\right)^2}\right)\left(\frac{1-e^{-(\beta + \alpha \mathbb{E}(J))\delta}}{\left(\beta + \alpha \mathbb{E}(J)\right)\delta}\right)\right),
\end{equation}
and moreover, in this stationary case,
\begin{align*}
C^{micro} &= \frac{2\mu_0\mathbb{E}(J^2)}{1 - \frac{\alpha \mathbb{E}(J)}{\beta}}, \\
C^{macro} &= \frac{2\mu_0 \mathbb{E}(J^2)}{\left(1 + \frac{\alpha \mathbb{E}(J)}{\beta}\right)^2 \left(1-\frac{\alpha \mathbb{E}(J)}{\beta}\right)}.
\end{align*}
The structure of the signature plot is the same as the one computed by Bacry et al.~\cite{Bacry2013a}, the main difference being the multiplicative term $\mathbb{E}(J^2)$ accounting for the random jump size and $\alpha$ which is multiplied by $\mathbb{E}(J)$ accounting for the presence of the jump size in the intensity.

\medskip

\subsection{Estimation}
\label{sec:estimation}
Let us estimate the parameters of the model on the data, using likelihood maximization. Observing continuously the price process $f$ corresponds to a continuous observation of the process $(N^+, N^-)$ and of the jump sizes $(J^+, J^-)$ at jump times. If $N$ is a Poisson process with intensity $\lambda$, the log-likelihood is equal to (see Daley and Vere-Jones~\cite[Proposition 7.2III]{daley2003}) 
\[\mathcal{L} = \int_0^ T \log(\lambda_t)dN_t + \int_0^T (1 - \lambda_t)dt.\]
The log-likelihood for one observation (that is, for one trading session) is then equal to the sum of
\[
    \mathcal{L}^+ = \sum_{i=1}^{N^+_T}\log\Big(\mu_0 e^{\kappa \frac{\tau_i^+}{T}} + \sum_{j=1}^{N^-_{\tau_i^+}}  \alpha J^-_j e^{-\beta \left(\tau_i^+ - \tau^-_j\right)}\Big) + T - \frac{\mu_0 T}{\kappa} \left(e^{\kappa} - 1\right) - \sum_{i=1}^{N^-_T}\frac{\alpha}{\beta}J^-_i \left(1-e^{-\beta\left(T-\tau^-_i\right)}\right)
\]
and
\[
    \mathcal{L}^- = \sum_{i=1}^{N^-_T}\log\Big(\mu_0 e^{\kappa \frac{\tau_i^-}{T}} + \sum_{j=1}^{N^+_{\tau_i^-}}  \alpha J^+_j e^{-\beta \left(\tau_i^- - \tau^+_j\right)}\Big) + T - \frac{\mu_0T}{\kappa} \left(e^{\kappa} - 1\right) - \sum_{i=1}^{N^+_T}\frac{\alpha}{\beta}J^+_i \left(1-e^{-\beta(T-\tau^+_i)}\right).
\]
If we dispose of continuous independent observations of the price process $f$ (and equivalently of observations of $N^+$, $N^-$ and $J$ at jump times) on $\left[0, T\right]$, the log-likelihood is equal to the sum of the likelihood of each observation. We perform the estimation on the whole dataset (prices between July and September) for maturities 18h, 19h and 20h and considering data from 9 hours to 1 hour before maturity (we then assume that $T = 8$ hours). We do not fit a model for the jump sizes but only estimate their first two moments by using empirical averages. To initialize the parameters of the minimization, except $\kappa$, we use estimates from the minimization of the distance between the averaged empirical signature plot and the signature plot associated to the stationary model when $\kappa = 0$, given by Equation \eqref{eq:signaturestati} for $\delta = 1, \ldots, 300$. The initial value for $\kappa$ is chosen equal to 0.1. Parameters are given in Table~\ref{tab:parameters}. Values for $\kappa$ are very close to each other, which could be interesting for a multidimensional model in term of parameter numbers.

\begin{table}
    \centering
    \begin{tabular}{ccccccc}
    \toprule
        Maturity & $\mu_0 \; (h^{-1})$ & $\kappa$ & $\alpha \; (h^{-1})$ & $\beta \; (h^{-1})$ & $\mathbb{E}(J)$ (\euro/MWh) & $\mathbb{E}(J^2)$ (\euro/MWh$^2$) \\
        \midrule
         18h & 2.49 & 3.51 & 864.39 & 237.30 & 0.13 & 0.066 \\
        19h &3.01 & 3.50& 2344.97 & 639.64 &0.13 &0.061\\
        20h & 3.06& 3.51 &3100.46 &859.11 & 0.13 &0.058\\
        \bottomrule
    \end{tabular}
    \caption{Estimated parameters for different maturities}
    \label{tab:parameters}
\end{table}

\subsection{Analysis} First, simulated prices are given in Figure \ref{fig:simu} for maturities 18h, 19h and 20h using parameters of Table \ref{tab:parameters} and jumps simulated from the empirical distributions. We also plot a price series from the dataset in each sub-figure. Prices are simulated using the well known thinning algorithm for Poisson processes given in \cite{ogata81}. It is difficult to distinguish simulations from the real prices and, at first sight, the model seems to reproduce the different stylized facts of the price, in particular the increase of market activity over the trading session. This intuition is confirmed by Figure \ref{fig:momentsimu} where empirical moments are compared to theoretical ones for the different considered maturities. The model succeeds in reproducing the shape of the different moments of the model with a low number of parameters. The expectation that reproduces the market activity can sometimes be slightly underestimated or overestimated. Of course one could improve these curves by estimating the baseline non parametrically but it would increase the complexity in the model a lot. Our main concern is the reproduction of the signature plot $C(t,\delta)$, $t \in \left[0,T\right]$, $\delta > 0$. Empirical signature plot (average of the different signature plot for the different trading days) and theoretical one are given in Figure~\ref{fig:signatureplotsimu} for different times. For every considered time, one observes the usual shape of the signature plot that exists in classic financial markets, with a very strong value at high frequency, a strong decrease then a stabilization. The model reproduces this shape very well. At last date $T$, the theoretical values are very close to the empirical ones. An interesting fact is the translation of the signature plot when time increases and gets closer to maturity : the whole curve goes upwards. For maturities 19h and 20h, theoretical curves values are above the empirical ones when time is equal to $T$ minus one or two hours. As at very high frequency, the signature plot is proportional to the expectation of the integrated intensity, see Equation \eqref{eq:Cmicro}, equal to $\frac{\mathbb{E}(f^+_t)}{\mathbb{E}(J)}$ which is overestimated for those two maturities at the end, see Figure~\ref{fig:momentsimu}. To conclude, our model represents the different empirical facts described in Section~\ref{sec:stylizedfact} well while being tractable and providing formulas for different quantities of interest for a practitioner. The model could be improved for instance by considering a full matrix for $\varphi$, or a more complex baseline $\mu$ however losing the analytical results.

\begin{figure}
    \centering
    \begin{subfigure}{0.60\textwidth}
    \centering
    \includegraphics[width=\textwidth]{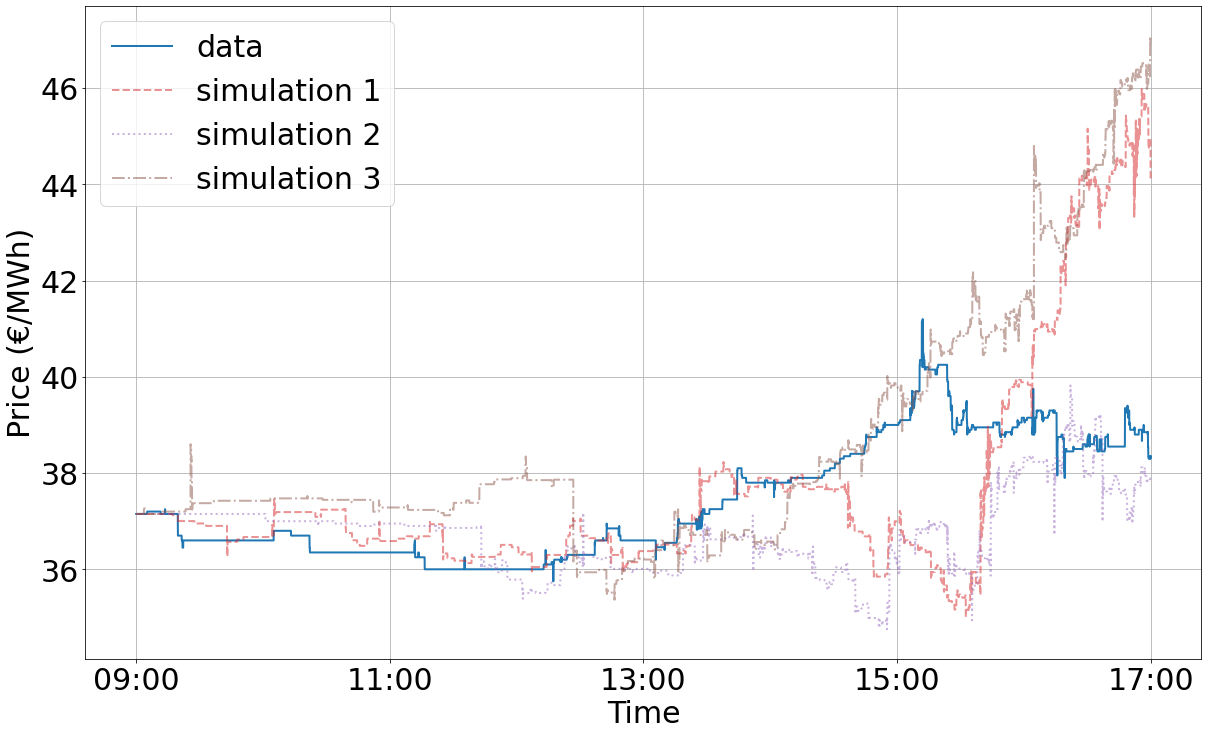}
    \caption{18h}
\end{subfigure}
   \begin{subfigure}{0.60\textwidth}
   \includegraphics[width=\textwidth]{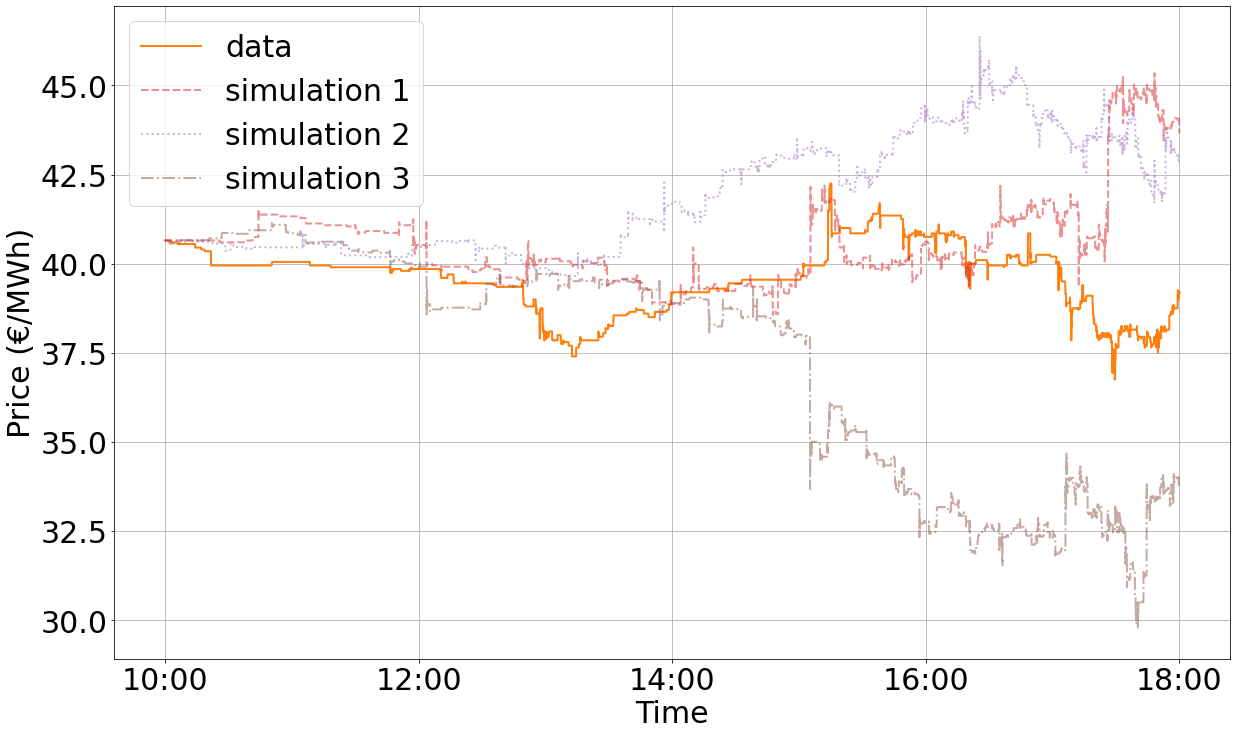}
    \caption{19h}
    \end{subfigure}
   \begin{subfigure}{0.60\textwidth}
   \includegraphics[width=\textwidth]{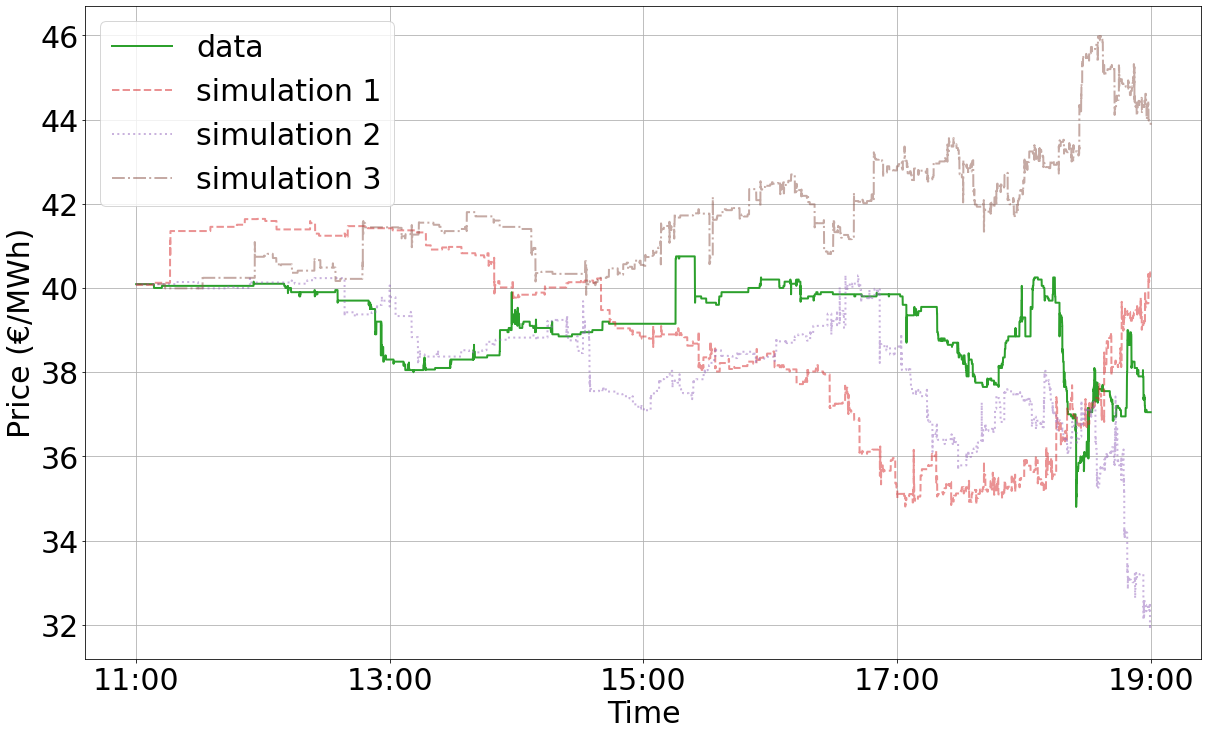}
    \caption{20h}
    \end{subfigure}
       \caption{\label{fig:simu} Price simulations, starting from the same initial value, for maturities 18h, 19h and 20h with estimated parameters in Table \ref{tab:parameters}, together with one sample}
\end{figure}

\begin{figure}
    \begin{subfigure}{0.49\textwidth}
    \centering
    \includegraphics[width=\textwidth]{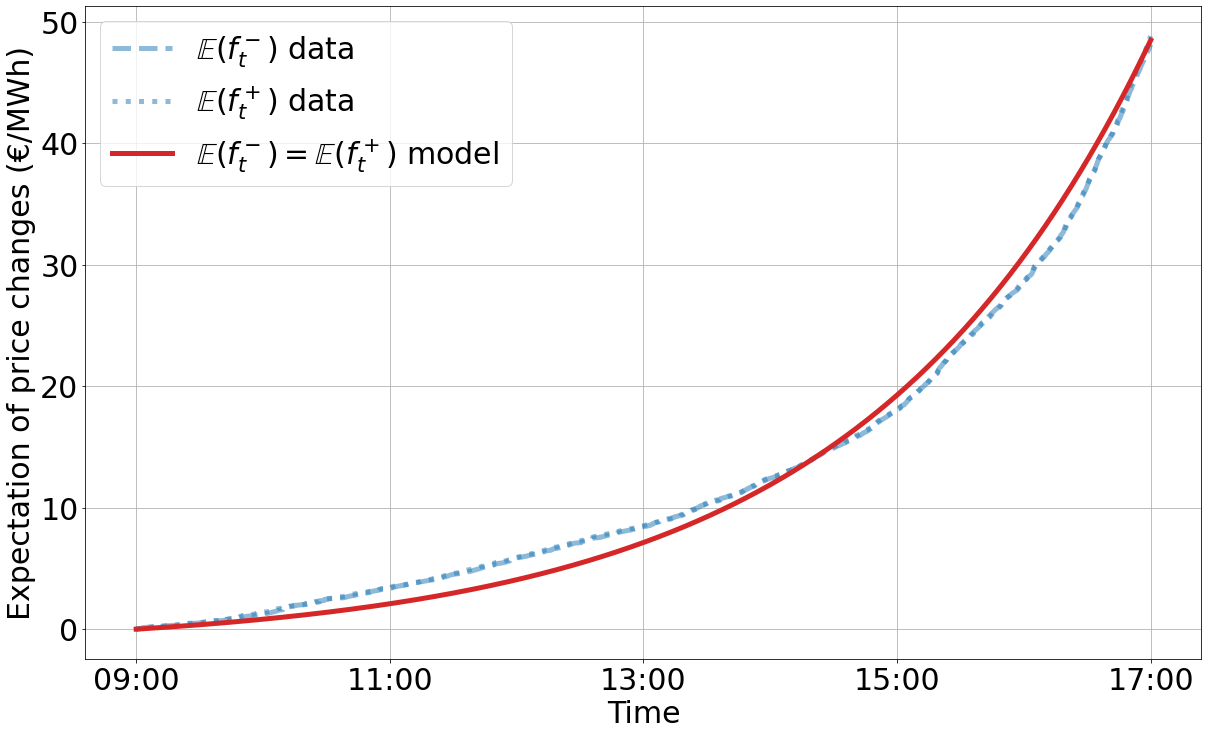}
    \caption{Expectation, 18h}
\end{subfigure}
    \begin{subfigure}{0.49\textwidth}
    \centering
    \includegraphics[width=\textwidth]{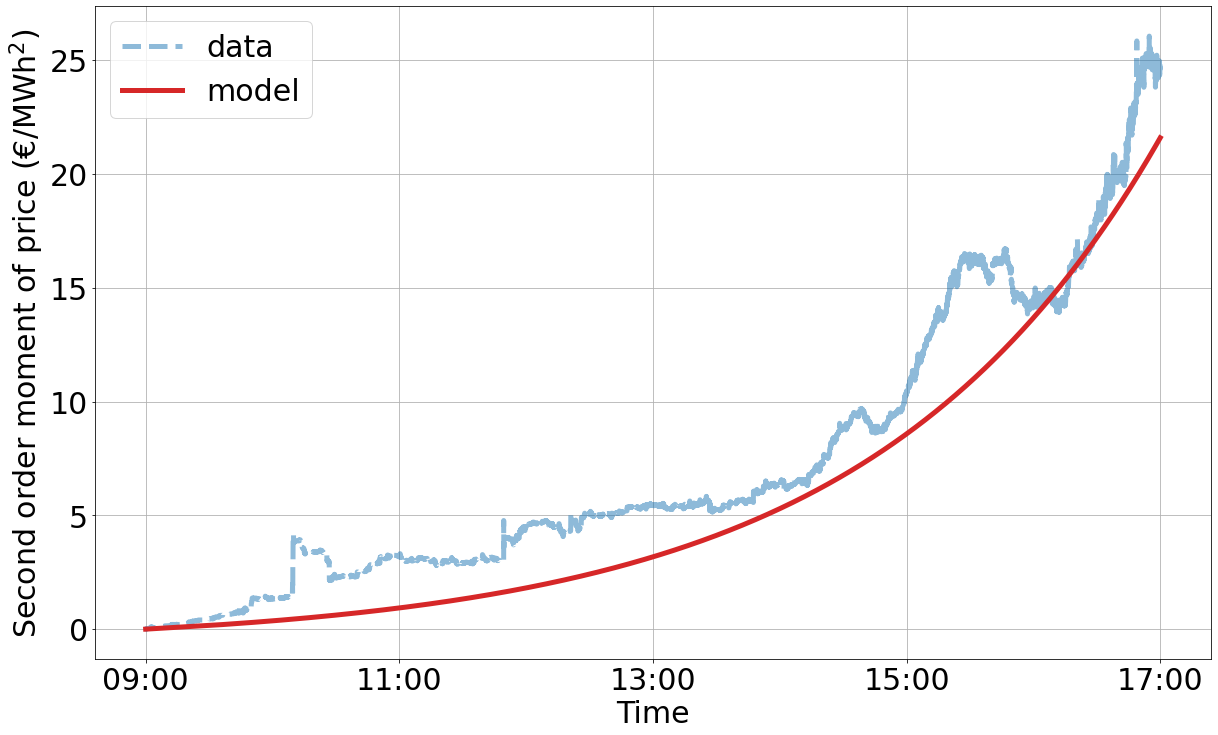}
    \caption{Second order moment, 18h}
\end{subfigure}
   \begin{subfigure}{0.49\textwidth}
   \includegraphics[width=\textwidth]{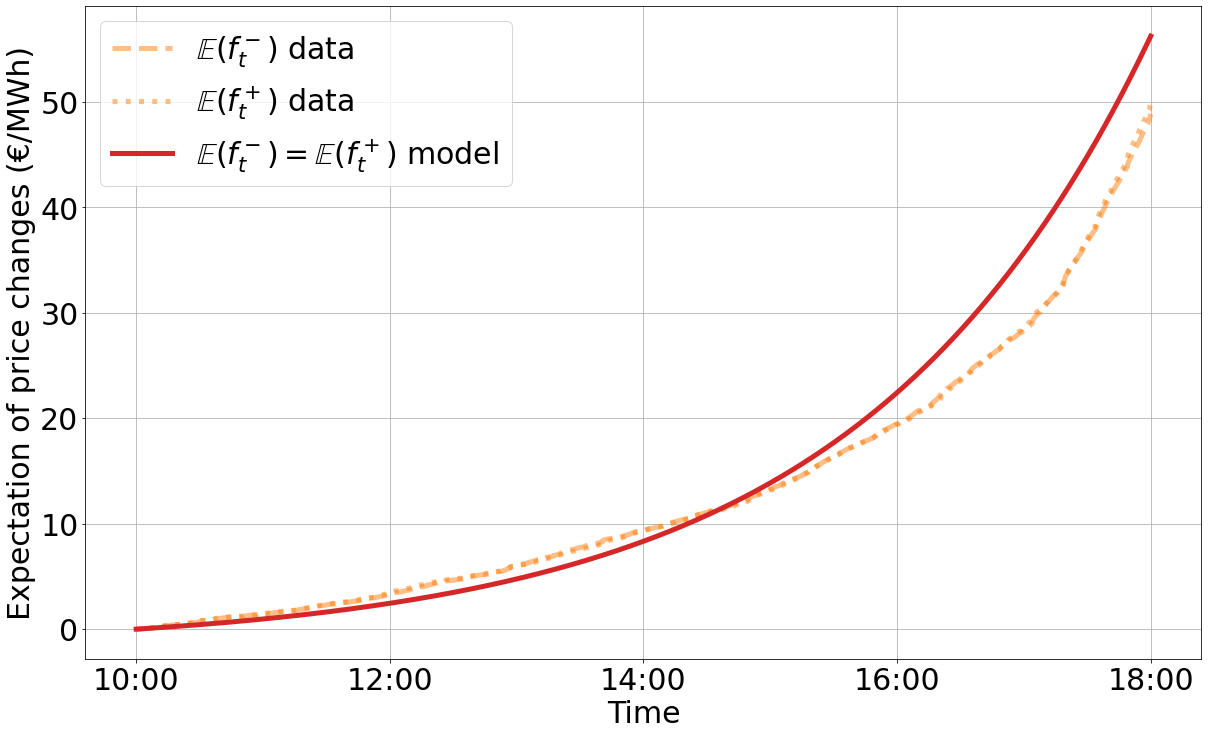}
    \caption{Expectation 19h}
    \end{subfigure}
        \begin{subfigure}{0.49\textwidth}
    \centering
    \includegraphics[width=\textwidth]{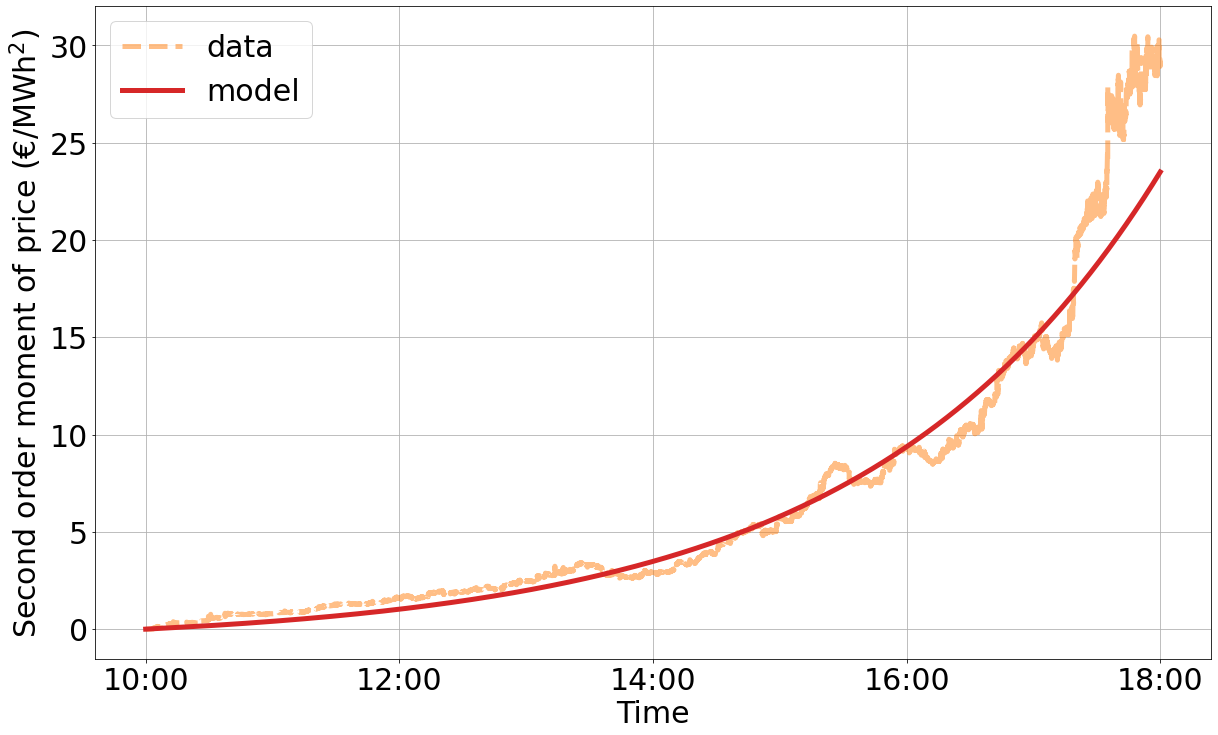}
    \caption{Second order moment, 19h}
\end{subfigure}
   \begin{subfigure}{0.49\textwidth}
   \includegraphics[width=\textwidth]{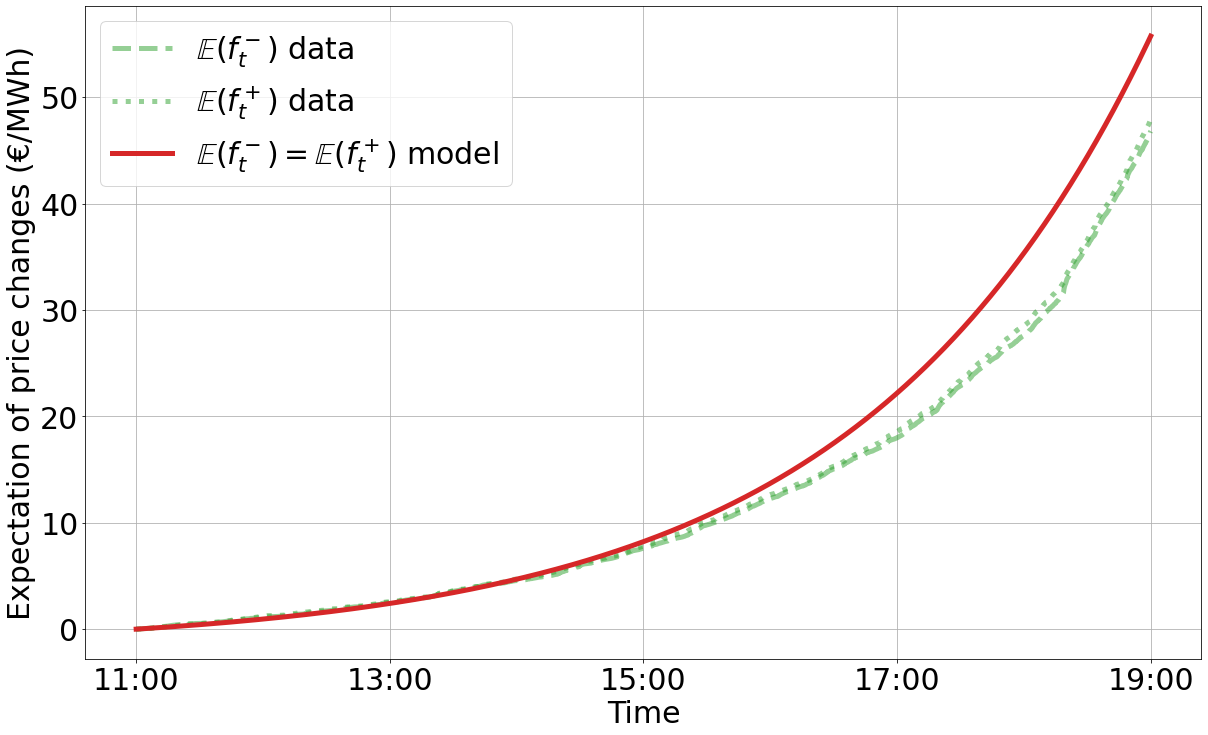}
    \caption{Expectation, 20h}
    \end{subfigure}
    \begin{subfigure}{0.49\textwidth}
    \centering
    \includegraphics[width=\textwidth]{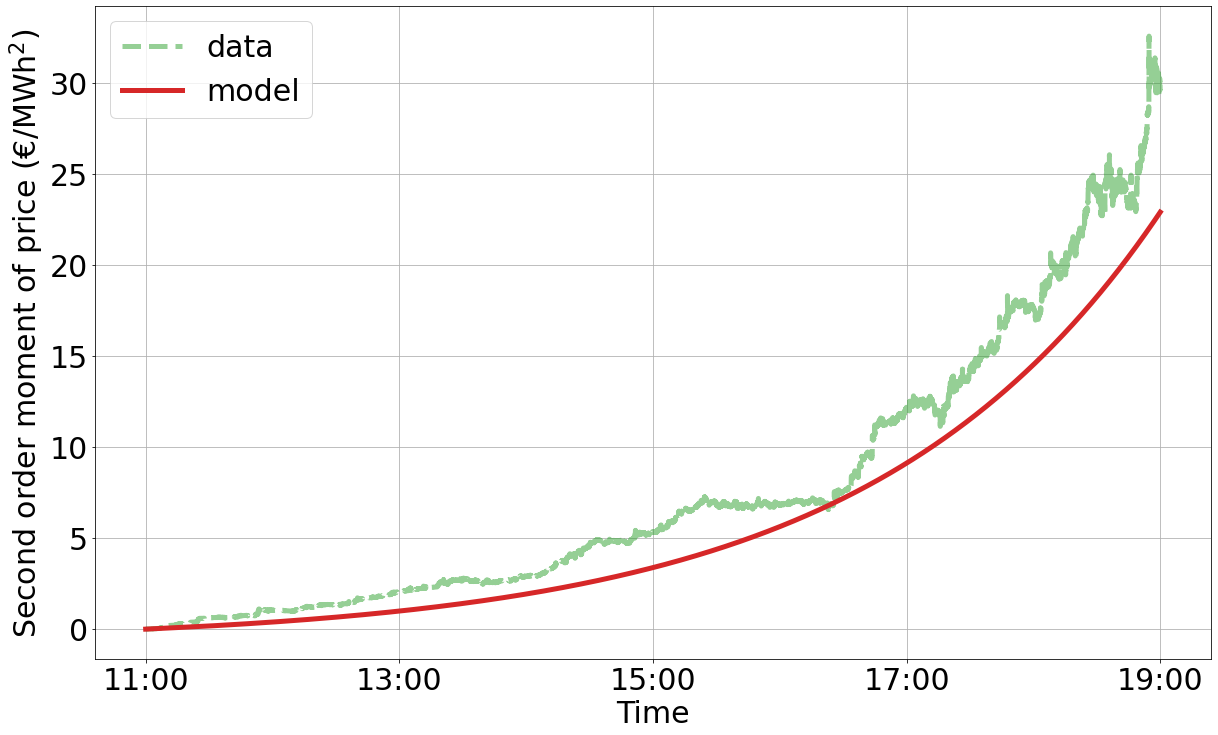}
    \caption{Second order moment, 20h}
\end{subfigure}
       \caption{\label{fig:momentsimu} Empirical and theoretical expectation of $f^+$ and $f^-$ and second order moment for price and for maturities 18h, 19h and 20h with estimated parameters in Table \ref{tab:parameters}}
\end{figure}

\begin{figure}
    \centering
    \begin{subfigure}{0.6\textwidth}
    \centering
    \includegraphics[width=\textwidth]{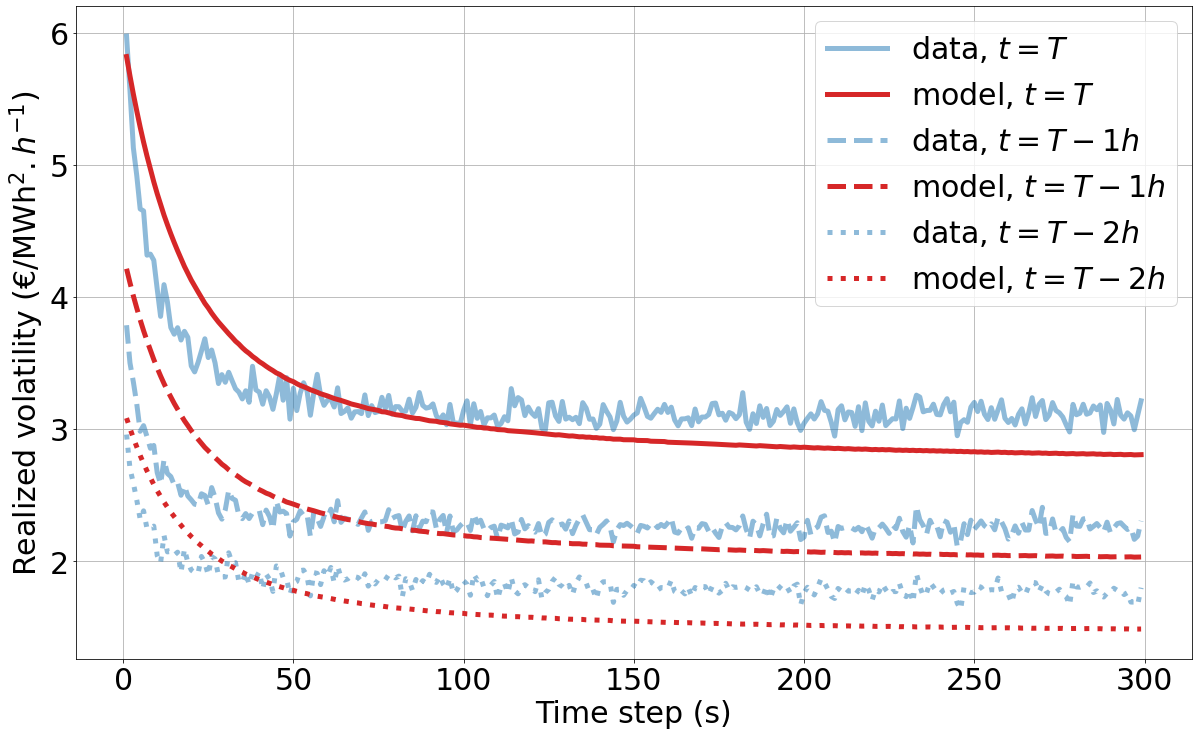}
    \caption{18h}
\end{subfigure}
   \begin{subfigure}{0.6\textwidth}
   \includegraphics[width=\textwidth]{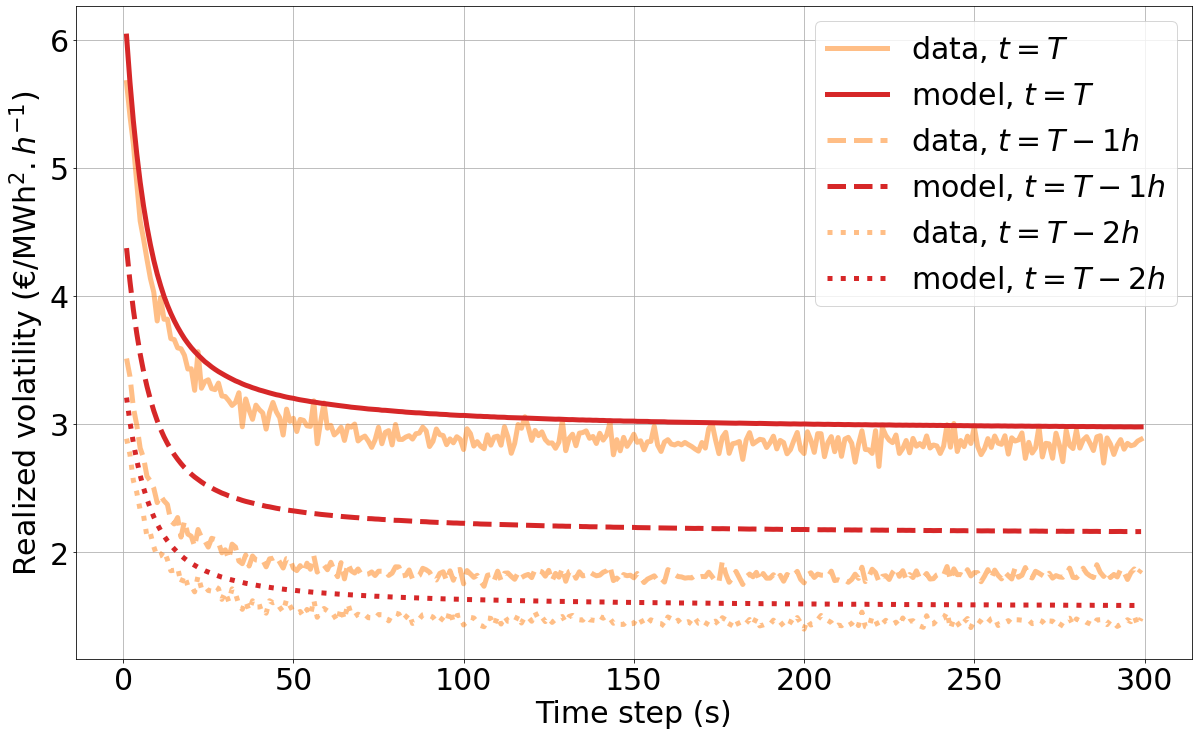}
    \caption{19h}
    \end{subfigure}
   \begin{subfigure}{0.6\textwidth}
   \includegraphics[width=\textwidth]{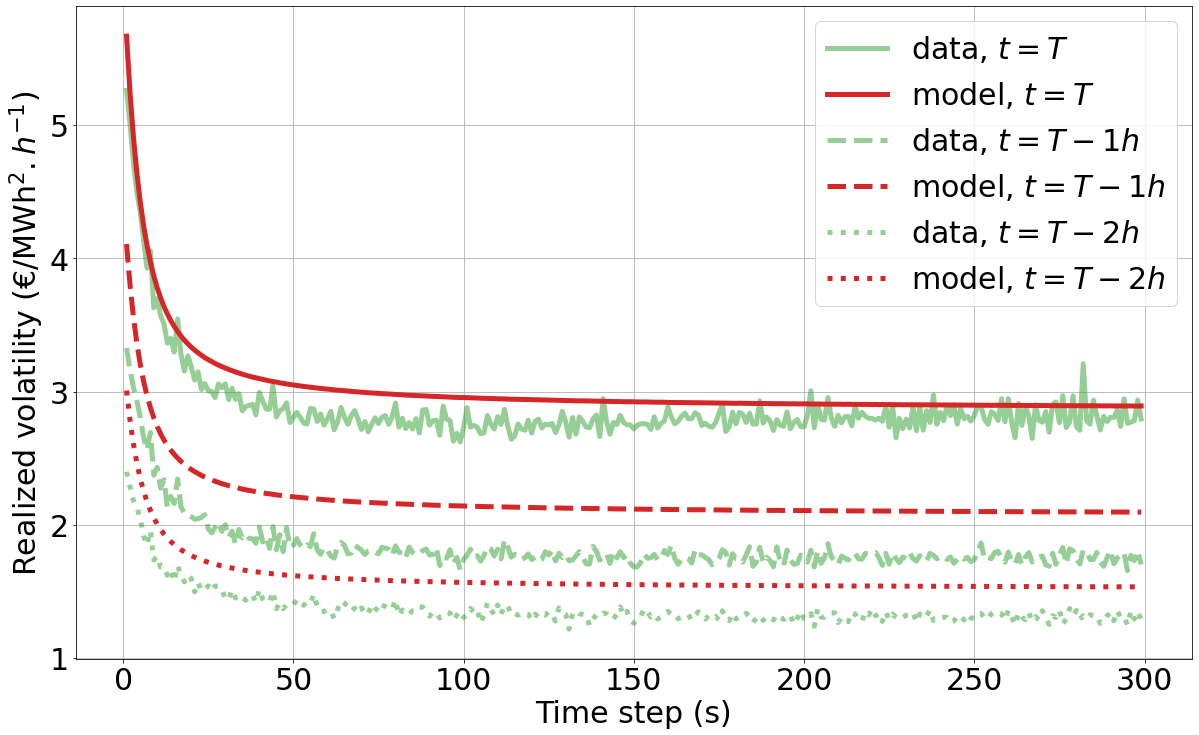}
    \caption{20h}
    \end{subfigure}
       \caption{\label{fig:signatureplotsimu}  Empirical and theoretical signature plot $C(t,\delta)$ at different times for maturities 18h, 19h and 20h with estimated parameters of Table~\ref{tab:parameters}}
\end{figure}

\section{Price at macroscopic scale and Samuelson effect}
\label{sec:macro}

In this section, we aim at defining the behavior of prices at large scale. Precisely, as done by Bacry et al.~\cite{Bacry2013b}, we provide a law of large numbers and a functional central limit theorem in our setting, as $T$ is sent to infinity. In the context of intraday trading sessions, of which duration is finite, this should be understood as having trading sessions of arbitrary length. The analysis of the prices at macroscopic scale will allow us to observe the Samuelson effect on intraday markets. In order to provide an adequate time normalization, the time index $t$ will be denoted as $vT$, where $v\in[0,1]$.\par
First we establish a law of large numbers for the sums of jump sizes to the power $0$, $1$ and $2$ respectively. Its proof is in Section~\ref{proof:proplln}.

\begin{proposition} \label{prop:lln} Let us consider the model \eqref{eq:intensities}, let $N^i_t = \int_0^t (J_s)^i dN_s = \left(\sum_{j=1}^{N^+_t} (J^+_j)^i, \sum_{j=1}^{N^-_t} (J^-_j)^i \right)^{\top}$ for $i \in \{0, 1, 2\}$. For $i \in \{0,1,2\}$, if $\mathbb{E}(J^{2i}) < \infty$,
\[\sup_{v \in \left[0,1\right]} \| T^{-1} N^i_{vT} - (I_2 - K)^{-1} \begin{pmatrix}
1 \\ 1
\end{pmatrix} \mathbb{E}(J^i) \int_0^v \mu(s) ds \| \to 0\] as $T \to \infty$ almost-surely and in $L^2(\mathbb{P})$.
\end{proposition}

This proposition gives the uniform convergence of $N_{\cdot T}^i$ to a limit function of $v$, featuring the integral of the baseline intensity. In this respect, it extends Theorem~1 of Bacry et al.~\cite{Bacry2013b} in which the baseline intensity is constant and the jump sizes are equal to 1. Now we provide a functional central limit theorem for the process $N^1$ met in Proposition~\ref{prop:lln}, which is the couple $(f^+, f^-)^{\top}$. Under an additional integrability condition, it gives the convergence in law of the quantity in Proposition~\ref{prop:lln} after a suitable renormalization. The convergence holds for the Skorokhod topology, see Skorokhod~\cite{Skorokhod1956} for the exact definition. The proof is in Section~\ref{proof:propbrownianlimit}.

\begin{proposition}\label{prop:brownianlimit}
Let us consider the model \eqref{eq:intensities}--\eqref{eq:upwarddownward}, let $N^1_t = (f_t^+, f_t^-)^{\top}$ and assume $\mathbb{E}(J^4) < \infty$. Let $\Sigma$ denote the diagonal matrix with $\Sigma_{j,j} = \left(\left(I_2-K\right)^{-1}\left(1,1\right)^{\top}\right)_j$, $j=1,2$. We have
\[\left(\frac{1}{\sqrt{T}}\left(N^1_{vT} - \mathbb{E}\left(N^1_{vT}\right)\right)\right)_{v \in \left[0,1\right]} \to \left((I_2-K)^{-1}\Sigma^{1/2} \sqrt{E(J^2)} \int_0^v \sqrt{\mu(s)} dW_s\right)_{v \in \left[0,1\right]}
\]
in law for the Skorokhod topology when $T \to \infty$ where $W$ is a 2-dimensional Brownian motion.
Adding the assumption $\int_0^{\infty} \varphi(t) t^{\frac{1}{2}} dt < \infty$ componentwise, we have
\[
\left(\sqrt{T}\left(\frac{N^1_{vT}}{T} - \left(I_2-K\right)^{-1}\begin{pmatrix}
1 \\ 1
\end{pmatrix} \int_0^v \mu(s) ds\right)\right)_{v \in \left[0,1\right]} \to \left((I_2-K)^{-1}\Sigma^{1/2} \sqrt{E(J^2)} \int_0^v \sqrt{\mu(s)} dW_s\right)_{v \in \left[0,1\right]}
\]
in law for the Skorokhod topology when $T \to \infty$.
\end{proposition}

Finally, under the same assumptions as in the above Proposition~\ref{prop:brownianlimit}, we can state the limit law of the price process in Corollary \ref{cor:pricelimit}.

\begin{corollary} \label{cor:pricelimit}
Let us consider the model \eqref{eq:intensities}--\eqref{eq:upwarddownward}--\eqref{eq:prices} under Assumption \ref{ass:parametric} and assume $\mathbb{E}(J^4)<\infty$. We have
\[\left(\frac{1}{\sqrt{T}}\left(f_{vT}-f_0\right)\right)_{v\in \left[0,1\right]} \to \left(\sqrt{\frac{2 E(J^2)}{\left(1+\frac{\alpha\mathbb{E}(J)}{\beta}\right)^2\left(1-\frac{\alpha \mathbb{E}(J)}{\beta}\right)}} \int_0^v \sqrt{\mu(s)} dW_s\right)_{v\in \left[0,1\right]}\] in law for the Skorokhod topology when $T \to \infty$, where $W$ is a 1-dimensional Brownian motion.
\end{corollary}
The price process converges to a Brownian motion with a time dependent volatility at large scale : this volatility increases at the same speed as the baseline intensity when time gets closer to maturity. The instantaneous squared volatility of the Brownian limit is equal to 
\[
(\sigma^{macro})^2(t) = \frac{2\mathbb{E}(J^2)\mu_0e^{\kappa t}}{\left(1+\frac{\alpha\mathbb{E}(J)}{\beta}\right)^2\left(1-\frac{\alpha\mathbb{E}(J)}{\beta}\right)}
\]
which is consistent with the macroscopic realized volatility given in Equation~\eqref{eq:macro}. 

\section{Conclusion and perspectives}
\label{sec:conclusion}

In this article, we first state some empirical facts about electricity intraday markets. In particular, we highlight an increasing price changing activity and the presence of microstructure noise through the empirical signature plot. Numerical illustrations on the German market are provided.\par
A simple model based on marked Hawkes processes is introduced. Generalizing the results of \cite{Bacry2013a}, we provide theoretical properties within our model, that allow us to reproduce the stylized facts quite well. In particular, such modeling allows us to fit the empirical signature plot. Similarly to electricity forward markets on which volatility increases as time gets closer to maturity, one also observes a Samuelson effect than can be reproduced within our model, in which the whole signature plot curve increases when approaching maturity.
Finally, we examine the behavior of prices at macroscopic scale that confirms the existence of this Samuelson effect : the price process converges to a Brownian motion with increasing instantaneous volatility.\par
Future research on this subject might focus on a multidimensional modeling with a specific focus on the Epps effect, according to which the correlation between two prices increases when the sampling frequency of estimation increases, as illustrated in Figure~\ref{fig:epps}. This is a well known effect in classic finance and would deserve being investigated in electricity intraday markets. We also think it would be meaningful to lead a large-scale study on the sizes of jumps in order to be able to use an even more representative distribution.

\begin{figure}
    \centering
     \centering
    \begin{subfigure}{0.49\textwidth}
    \centering
    \includegraphics[width=\textwidth]{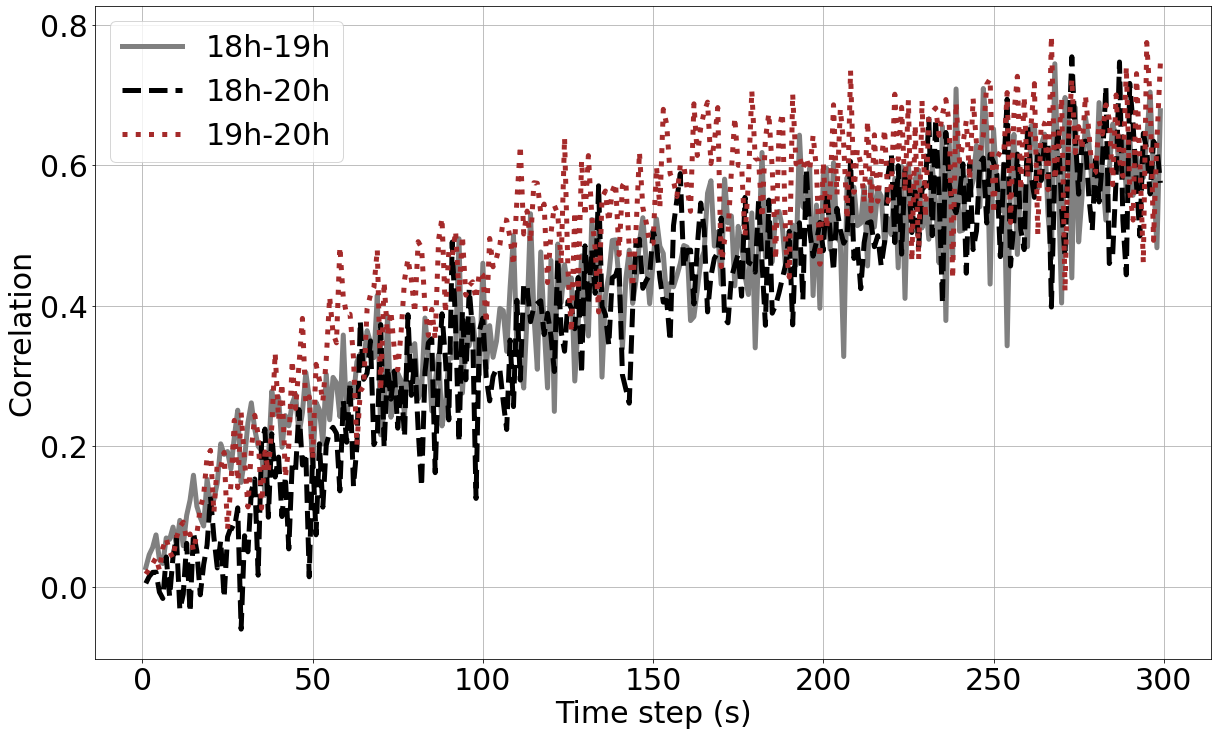}
    \caption{August 30\textsuperscript{th}, 2017}
\end{subfigure}
   \begin{subfigure}{0.49\textwidth}
    \centering
   \includegraphics[width=\textwidth]{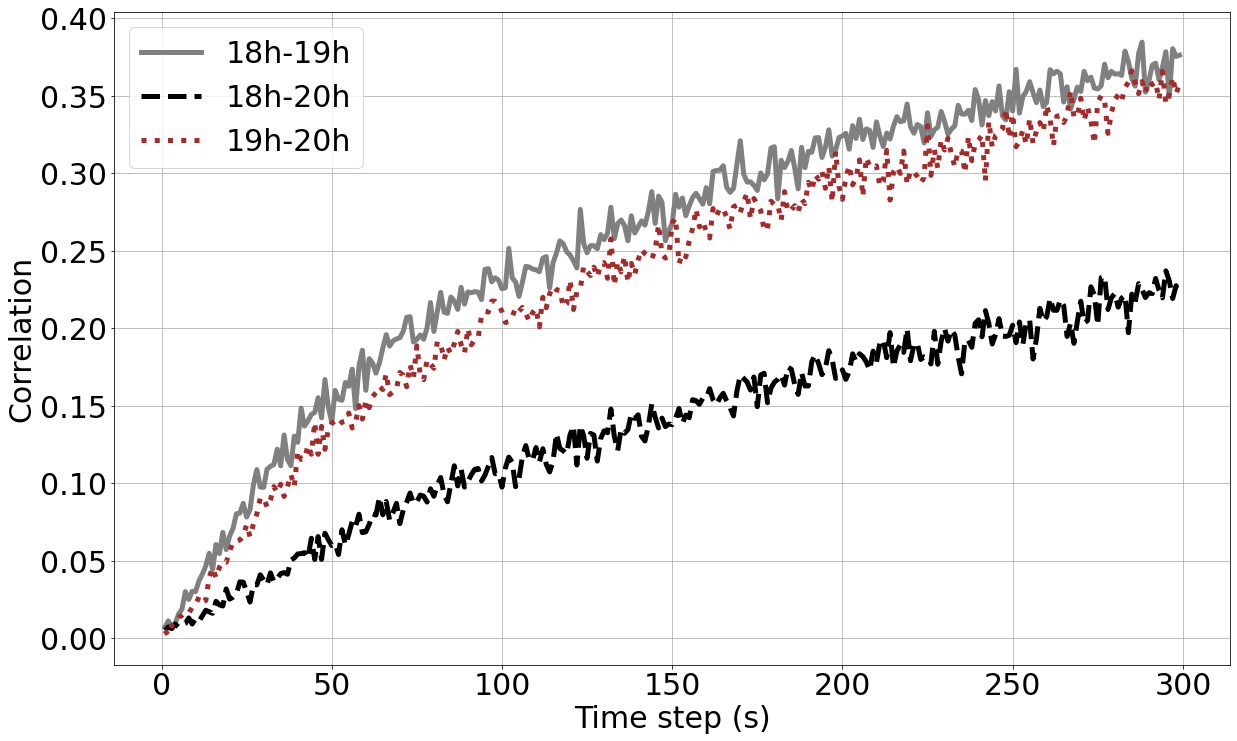}
    \caption{Average}
    \end{subfigure}
       \caption{\label{fig:epps} Epps effect for different maturities estimated from 9 hours before the nearest maturity for one trading session and on average over all trading sessions}
\end{figure}

\section{Proofs for Section \ref{sec:model}}
\label{proof:sectionmodel}
\subsection{Preliminary results}

The following two results are required for the computation of the different moments. The first result, in Proposition~\ref{prop:charachawkes}, gives the characteristic function of the multivariate marked Hawkes process~\eqref{eq:upwarddownward}. The moments will be derived from this function.  

\begin{proposition}\label{prop:charachawkes} Let $(f^+, f^-)^{\top}$ defined the bivariate marked Hawkes process defined by Equation~\eqref{eq:upwarddownward}. Its characteristic function 
\[L(a_+, a_-, t) = \mathbb{E}\left(e^{ia_+ f^+_t + ia_- f^-_t}\right)\]
for $t \in \left[0,T\right]$, $a_+ \in \mathbb{R}$, $a_- \in \mathbb{R}$, is given by
\[
L(a_+, a_-, t) = \exp\left(\int_0^t \mu\left(\frac{s}{T}\right)\begin{pmatrix}1 \\ 1 \end{pmatrix}^{\top}\left(C(a_+, a_-, t-s) -\begin{pmatrix}1 \\ 1 \end{pmatrix}\right)ds\right)
\]
with $C$ solution of the integro-differential equation 
\[
C(a_+, a_-, t) = \mathbb{E}\left(\exp\left(iJ \begin{pmatrix}a_+ \\ a_- \end{pmatrix} + J \int_0^t \varphi(s)^{\top} \left(C(a_+,a_-,t-s) - \begin{pmatrix}1 \\ 1 \end{pmatrix} \right)ds\right)\right).
\]
\end{proposition}

\begin{proof}
The proof is similar to the one of \cite[Proposition 5]{Gao2018} in a multidimensional setting. It also extends the one of \cite[Theorem 3.1]{ElEuch19} including marks. Let us write
\[\varphi = \begin{pmatrix}
\varphi_{+, +} & \varphi_{+, -} \\
\varphi_{-, +} & \varphi_{-, -}
\end{pmatrix}.
\]
In the following, we consider the cluster representation of Hawkes processes introduced in \cite{hawkes1974}. The first generation of migrants with type $+$ and the first generation of migrants with type $-$ both appears as an inhomogeneous Poisson process with intensity $\mu(\frac{\cdot}{T})$. Each new migrant is associated to a random jump mark with law $J$, and the marks are i.i.d. Migrants of type $+$ (resp. $-$) of the first generation give birth to migrants of the second generation of type $+$ as a inhomogeneous Poisson process with intensity $x\varphi_{+,+}(\cdot-t)$ (resp. $x\varphi_{+,-}(\cdot-t$) and of type $-$ with intensity $x\varphi_{-,+}(\cdot-t)$ (resp. $x\varphi_{-,-}(\cdot-t)$) if migrant arrived at time $t$ and is associated to a mark $x > 0$. Migrants of the second generation give birth to migrants of the third generation in the same way, and so on. $N^+_t$ and $N^-_t$ correspond respectively to the total number of migrants of type $+$ and of type $-$ up to time $t$ and $f^+_t$ and $f_t^-$ are the sum of their marks up to time $t$. 

\medskip
For $m \in \{+,-\}$ and a given mark $x > 0$, one considers the point process $\left(\left(\tilde{N}^{m,+}_{x,t}, \tilde{N}^{m,-}_{x,t}\right)^{\top}\right)_t$ defined with a first generation that arrives as an inhomogeneous Poisson process with rate $x\varphi_{+,m}$ for positive migrants and $x\varphi_{-,m}$ for negative migrants, and migrants are associated with an i.i.d. sequence with same law as $J$ that represents their mark. Migrants of type $j \in \{+,-\}$ with mark $y > 0$ born at time $t$ give birth to children of type $l \in \{+,-\}$ as an inhomogeneous Poisson process with intensity $y\varphi_{l,j}(\cdot-t)$. Each born migrant is associated with a mark having same law as $J$, and the marks are independent. The next generations are created in the same way as the second one. Let $\tilde{f}^{m,+}$ and $\tilde{f}^{m,-}$ be respectively the sum of the mark of all positive and negative migrants. Let
\[C_m(a_+,a_-,t) = \mathbb{E}\left(\exp\left(ia_m J + i a_+ \tilde{f}_{J,t}^{m,+} +i a_- \tilde{f}_{J,t}^{m,-} \right)\right)\]
for $a_+, \;a_- \in \mathbb{R}$ and $t \in \left[0,T\right]$ be the characteristic function of $\left(J{\bf 1}_{m=+} + \tilde{f}^{m,+}_{J,t},J{\bf 1}_{m=-} + \tilde{f}^{m,-}_{J,t}\right)^{\top}$ for $m \in \{+,-\}$ and let us consider the conditional version 
\[C^c_m(a_+,a_-,t, x) = \mathbb{E}\left(\exp\left(ia_m J + i a_+ \tilde{f}_{J,t}^{m,+} +i a_- \tilde{f}_{J,t}^{m,-} \right) | J = x\right).\]

\medskip
Let $N^{+,0}_t$ and $N^{-,0}_t$ be the number of migrants of type $+$ and of type $-$ of the first generation at time $t$, that have arrived at times $(\tau_i^{+,0})$ and $(\tau_i^{-,0})$ and with marks $(J^{+,0}_i)$ and $(J^{-,0}_i)$. The sum of the marks of migrants of type $+$ (resp. $-$) is then equal to the sum of $f^{+,0}_t$ (resp. $f^{-,0}_t$), which is the sum of the marks of first generation, and of the sum of the marks of positive migrants from the second generation. For a first generation migrant of type $m \in \{+,-\}$ born at time $s \leq t$ and with mark $x$, the sum of the marks of positive (resp. negative) migrants induced by its progeny has the same law as $\tilde{f}^{m,+}_{x,t-s}$ (resp. $\tilde{f}^{m,-}_{x,t-s}$). Therefore, we have the equality in law 
\begin{equation}
    \label{eq:laweq}
\begin{pmatrix}
f_t^+ \\ f_t^-
\end{pmatrix} = \begin{pmatrix}
f^{+,0}_t \\ f^{-,0}_t\end{pmatrix}
+ \sum_{i=1}^{N^{-,0}_t} 
\begin{pmatrix}
\tilde{f}^{-,+,i}_{J^{-,0}_i, t-\tau^{-,0}_i} \\
\tilde{f}^{-,-,i}_{J^{-,0}_i, t-\tau^{-,0}_i}
\end{pmatrix} 
+\sum_{i=1}^{N^{+,0}_t} 
\begin{pmatrix}
\tilde{f}^{+,+,i}_{J^{+,0}_i, t-\tau^{+,0}_i}  \\
\tilde{f}^{+,-,i}_{J^{+,0}_i, t-\tau^{+,0}_i}
\end{pmatrix}
\end{equation}
where, for $m=+,-$, for $i=1,...,N_t^{m,0}$, $(\tilde{f}^{m,+,i}_{J^{-,0}_i}, \tilde{f}^{m,-,i}_{J^{-,0}_i})^{\top}$ is an independent copy of $(\tilde{f}^{m,+}_{J^{-,0}_i},\tilde{f}^{m,-}_{J^{-,0}_i})^{\top}$, independent of $N^0=\left(N^{-,0}, N^{+,0}\right)^{\top}$.

\medskip
From \eqref{eq:laweq}, we deduce 
\[
\begin{split}
\mathbb{E}\left(e^{a_+ f_t^+ + a_- f_t^-} | f^0\right) &= e^{a_+ f^{+,0}_t + a_- f^{-,0}_t} \prod_{m \in \{+,-\}}\prod_{i=1}^{N^{m,0}_t} \mathbb{E}\left(\exp\left(a_+\tilde{f}^{m,+,i}_{J^{m,0}_i, t-\tau^{m,0}_i} + a_-\tilde{f}^{m,-,i}_{J^{m,0}_i, t-\tau^{m,0}_i}\right) | f_0\right) \\
&= \prod_{m \in \{+,-\}}\prod_{i=1}^{N_t^{m,0}} \mathbb{E}\left(\exp\left(ia_m J^{m,0}_i + a_+\tilde{f}^{m,+,i}_{J^{m,0}_i, t-\tau^{m,0}_i} + a_-\tilde{f}^{m,-,i}_{J^{m,0}_i, t-\tau^{m,0}_i}\right) | f_0\right) \\
&=\prod_{m \in \{+,-\}}\prod_{i=1}^{N^{m,0}_t} C^c_m(a_+,a_-,t-\tau^{m,0}_i,J^{m,0}_i).
\end{split}
\]
As the marks $(J^{m,0}_i)$ are independent,
\[
\mathbb{E}\left(e^{a_+ f_t^+ + a_- f_t^-} | N^0\right) = \prod_{m \in \{+,-\}}\prod_{i=1}^{N^{m,0}_t} C_m(a_+,a_-,t-\tau^{m,0}_i).
\]
For $m \in \{+,-\}$, conditionally on $N^{0,m}_t$, the vector of jump times $(\tau^{0,m}_1, \ldots, \tau^{0,m}_{N_t^{m,0}})^{\top}$ has the same law as the order statistics of i.i.d. random variables with density $s \mapsto \frac{\mu(\frac{s}{T}){\bf 1}_{s \leq t}}{\int_0^t \mu(\frac{s}{T})ds}$. We then have 
\[
\mathbb{E}\left(e^{a_+ f_t^+ + a_- f_t^-} | N^0_t\right) = \prod_{m \in \{+,-\}} \left(\int_0^t C_m(a_+,a_-,t-s) \frac{\mu\left(\frac{s}{T}\right)}{\int_0^t \mu\left(\frac{s}{T}\right)ds}ds\right)^{N^{0,m}_t}.
\]
Finally, as $N^{0,+}$ and $N^{0,-}$ are independent inhomogeneous Poisson processes, we get 
\[
\mathbb{E}\left(e^{a_+ f_t^+ + a_- f_t^-}\right) = \exp\left(\sum_{m \in \{+,-\}}\int_0^t \left(C_m\left(a_+,a_-,t-s\right) - 1\right)\mu\left(\frac{s}{T}\right)ds\right).
\]
Conditionally on a given mark $J$, $\left(\tilde{f}_J^{l,+}, \tilde{f}_J^{l,-}\right)^{\top}$ is also a bivariate marked Hawkes process defined in the same way as $\left(f^{+}, f^{-}\right)^{\top}$ but with baseline equal to $\left(J\varphi_{+,l}, J\varphi_{-,l}\right)^{\top}$ for $l \in \{+,-\}$. Therefore we can find, for $l \in \{+,-\}$,
\[\mathbb{E}\left(e^{a_+ \tilde{f}_J^{l,+} + a_- \tilde{f}_J^{l,-}} | J\right) = \exp\left(\sum_{m \in \{+,-\}}\int_0^t J(C_m(a_+,a_-,t-s) - 1)\varphi_{m,l}(s)ds\right)
\]
and $C_l(a_+,a_-,t)$ is equal to
\[
\mathbb{E}\left(e^{a_l J + a_+ \tilde{f}_J^{l,+} + a_- \tilde{f}_J^{l,-}}\right) = \mathbb{E}\left(\exp\left(a_l J + \sum_{m \in \{+,-\}}J\int_0^t (C_m(a_+,a_-,t-s) - 1)\varphi_{m,l}(s)ds\right)\right).
\]
\end{proof}

\bigskip
Second result is Lemma \ref{lemma:ode}. The equation in this lemma is a Volterra equation of the second kind that often appears in moments computation, and the proof follows directly from \cite[Lemma~3]{Bacry2013b}.

\begin{lemma} \label{lemma:ode} Let $a, \beta \in \mathbb{R}$, $|a| < |\beta|$, $\varphi: x \mapsto a e^{-\beta x}$ and $f: \mathbb{R}_+ \to \mathbb{R}$ some measurable and locally bounded function. 
The unique solution of 
\[
\Psi(t) = f(t) + \int_0^t \Psi(t-s) \varphi(s) ds 
\]
with unknown $\Psi$ is given by 
\[
\Psi(t) = f(t) + \int_0^t f(s) \varphi(t-s) e^{a(t-s)}ds.
\]
\end{lemma}

\subsection{Proof of Proposition \ref{prop:expectation}}
\label{proof:propexpectation}

Let $t \in \left[0,T\right]$. The expectations of $f^+_t$ and $f^-_t$ are respectively given by $-i\frac{\partial L}{\partial a_+}(0,0,t)$ and $-i\frac{\partial L}{\partial a_-}(0,0,t)$ with $L$ the characteristic function of $(f^+_t, f^+_t)^{\top}$. Applying Proposition \ref{prop:charachawkes} with $\varphi = \begin{pmatrix}
0 & \varphi_{\exp} \\
\varphi_{\exp} & 0 
\end{pmatrix}$, we have 
\begin{equation} \label{carac}
L\left(a_+, a_-, t\right) = \exp\left(\int_0^t \left(C_+\left(a_+, a_-, s\right) -1\right)\mu\left(\frac{t-s}{T}\right) ds + \int_0^t \left(C_-\left(a_+, a_-, s\right) -1\right)\mu\left(\frac{t-s}{T}\right)ds\right)
\end{equation}
with $C_+$ and $C_-$ solutions of 
\begin{equation} \label{c}
\begin{cases}
C_+\left(a_+, a_-, t\right) = \mathbb{E}\left(\exp\left(i J a_+ + \int_0^t J C_-\left(a_+,a_-,t-s\right)\varphi_{\exp}\left(s\right)ds\right)\right)\\
C_-\left(a_+, a_-, t\right) = \mathbb{E}\left(\exp\left(i J a_- + \int_0^t J C_+\left(a_+,a_-,t-s\right)\varphi_{\exp}\left(s\right)ds\right)\right)
\end{cases}.
\end{equation}
By differentiating \eqref{carac} and \eqref{c} with respect to $a_+$ and $a_-$, we find
\begin{equation} \label{eq:Mplus}
\mathbb{E}(f_t^+) = \int_0^t \Psi_+(s) \mu\left(\frac{t-s}{T}\right)ds,
\end{equation}
\begin{equation} \label{eq:Mminus}
\mathbb{E}(f_t^-) = \int_0^t \Psi_-(s) \mu\left(\frac{t-s}{T}\right)ds
\end{equation}
with $\Psi_+(t) = -i\frac{\partial (C^+ + C^-)}{\partial a_+}(0,0,t)$ and $\Psi_-(t) = -i\frac{\partial (C^+ + C^-)}{\partial a_-}(0,0,t)$
and we have the system of equations 
\begin{equation} \label{eq:systemC}
\left \{
\begin{array}{rcl}
-i \frac{\partial C^+}{\partial a_+}(0,0,t)&=&\mathbb{E}(J) + \int_0^t \mathbb{E}(J) \left(- i \frac{\partial C^-}{\partial a_+}(0,0,t-s)\right)\varphi_{\exp}(s)ds\\
-i \frac{\partial C^-}{\partial a_-}(0,0,t)&=&\mathbb{E}(J) + \int_0^t \mathbb{E}(J)\left(- i \frac{\partial C^+}{\partial a_-}(0,0,t-s)\right)\varphi_{\exp}(s)ds\\
-i \frac{\partial C^+}{\partial a_-}(0,0,t)&=&\int_0^t \mathbb{E}(J) \left(- i \frac{\partial C^-}{\partial a_-}(0,0,t-s)\right)\varphi_{\exp}(s)ds\\
-i \frac{\partial C^-}{\partial a_+}(0,0,t)&=&\int_0^t \mathbb{E}(J) \left(- i \frac{\partial C^+}{\partial a_+}(0,0,t-s)\right)\varphi_{\exp}(s)ds\\
\end{array}
\right.
\end{equation}
with solutions found using Lemma \ref{lemma:ode} :
\begin{equation} \label{eq:solutionC}
\left \{
\begin{array}{rcl}
-i \frac{\partial C^+}{\partial a_+}(0,0,t)&=&\frac{\beta \mathbb{E}(J)}{2(\beta - \alpha \mathbb{E}(J))} + \frac{\beta \mathbb{E}(J)}{2(\beta + \alpha \mathbb{E}(J))} -\frac{\alpha \mathbb{E}(J)^2}{2(\beta - \alpha \mathbb{E}(J))}e^{-(\beta - \alpha\mathbb{E}(J))t} + \frac{\alpha \mathbb{E}(J)^2}{2(\beta + \alpha \mathbb{E}(J))}e^{-(\beta + \alpha\mathbb{E}(J))t} \\
-i \frac{\partial C^-}{\partial a_-}(0,0,t)&=&-i \frac{\partial C^+}{\partial a_+}(0,0,t)\\
-i \frac{\partial C^+}{\partial a_-}(0,0,t)&=&\frac{\beta \mathbb{E}(J)}{2(\beta - \alpha \mathbb{E}(J))} - \frac{\beta \mathbb{E}(J)}{2(\beta + \alpha \mathbb{E}(J))} -\frac{\alpha \mathbb{E}(J)^2}{2(\beta - \alpha \mathbb{E}(J))}e^{-(\beta - \alpha\mathbb{E}(J))t} - \frac{\alpha \mathbb{E}(J)^2}{2(\beta + \alpha \mathbb{E}(J))}e^{-(\beta + \alpha\mathbb{E}(J))t}\\
-i \frac{\partial C^-}{\partial a_+}(0,0,t)&=&-i \frac{\partial C^+}{\partial a_-}(0,0,t)\\
\end{array}
\right..
\end{equation}

We then have
\begin{equation} \label{eq:solutionpsi}\Psi_+(t) = \Psi_-(t) = \frac{-\alpha\mathbb{E}(J)^2}{\beta - \alpha \mathbb{E}(J)} e^{-(\beta - \alpha \mathbb{E}(J)) t} + \frac{\beta\mathbb{E}(J)}{\beta - \alpha\mathbb{E}(J)}\end{equation}
and $\mathbb{E}(f^+_t)$, $\mathbb{E}(f^-_t)$ can be derived from \eqref{eq:Mplus} and \eqref{eq:Mminus}. The computation of $\mathbb{E}(\lambda_t^+)$ and $\mathbb{E}(\lambda_t^-)$ follows from 
\[\mathbb{E}(f_t^+) = \mathbb{E}(J) \mathbb{E}(\int_0^t \lambda^+_s ds)\]
and
\[\mathbb{E}(f_t^-) = \mathbb{E}(J) \mathbb{E}(\int_0^t \lambda^-_s ds).\]

\subsection{Proof of Proposition \ref{prop:moment2}}
\label{proof:propmoment2}

Let $t \in \left[0,T\right]$. The expectation of $(f^+_t)^2$, $(f^-_t)^2$ and $f^+_t f^-_t$ are respectively given by $-\frac{\partial^2 L}{\partial a_+^2}(0,0,t)$, $-\frac{\partial^2 L}{\partial a_-^2}(0,0,t)$ and $-\frac{\partial^2 L}{\partial a_+ \partial a_-}(0,0,t)$ with $L$ defined in \eqref{carac}--\eqref{c}. Thus, 
\[\mathbb{E}((f^+_t)^2) = \left(\int_0^t\Psi_+(s)\mu\left(\frac{t-s}{T}\right)ds\right)^2 + \int_0^t \Psi_{2,+}(s)\mu\left(\frac{t-s}{T}\right)ds \]
with $\Psi_+$ defined in \eqref{eq:solutionpsi} and $\Psi_2(t) =-\frac{\partial^2 (C^+ + C^-)}{\partial a_+^2}(0,0,t)$ solution of
\[
\begin{split}
\Psi_{2,+}(t) =& \mathbb{E}(J^2)\left(1 + \int_0^t -i\frac{\partial C^-}{\partial a_+}(0,0,t-s)\varphi_{\exp}(s)ds\right)^2 \\
&+ \mathbb{E}(J^2)\left(\int_0^t -i\frac{\partial C^+}{\partial a_+}(0,0,t-s)\varphi_{\exp}(s)ds\right)^2  + \int_0^t \Psi_{2,+}(t-s)\varphi_{\exp}(s)ds,
\end{split}
\]
which can be rewritten using \eqref{eq:systemC} :
\[\Psi_{2,+}(t) = \frac{\mathbb{E}(J^2)}{\mathbb{E}(J)^2}\left(-i\frac{\partial C^+}{\partial a_+}(0,0,t)\right)^2 + \frac{\mathbb{E}(J^2)}{\mathbb{E}(J)^2} \left(-i\frac{\partial C^-}{\partial a_+}(0,0,t)\right)^2  + \int_0^t \Psi_{2,+}(s)\varphi_{\exp}(t-s)ds.
\]
We easily show that 
\[\mathbb{E}((f^-_t)^2) = \mathbb{E}((f^+_t)^2).\]
The cross moment can be computed in the same way and it is equal to 
\[
\mathbb{E}(f^+_t f^-_t) = \left(\int_0^t\Psi_+(s)\mu\left(\frac{t-s}{T}\right)ds\right)^2 + \int_0^t \Psi_{2,+,-}(s)\mu\left(\frac{t-s}{T}\right)ds
\]
with $\Psi_{2,+,-}(t) = - \frac{\partial C^+ + C^-}{\partial a^+ \partial a^-}(0,0,t)$ solution of 
\[
\Psi_{2,+,-}(t) = 2\frac{\mathbb{E}(J^2)}{\mathbb{E}(J)^2} \left(-i\frac{\partial C^+}{\partial a_+}(0,0,t)\right) \left(-i\frac{\partial C^-}{\partial a_+}(0,0,t) \right) + \mathbb{E}(J)\int_0^t \Psi_{2,+,-}(s)\varphi_{\exp}(t-s)ds.
\]
Therefore, 
\begin{equation} \label{eq:moment2integral}
\begin{split}
\mathbb{E}(f_t^2) &= f_0^2 + 2\left(\mathbb{E}\left(\left(f_t^+\right)^2\right) - \mathbb{E}\left(f_t^+ f_t^-\right)\right) \\
&= f_0^2 + 2 \int_0^t \Psi_2(s) \mu\left(\frac{t-s}{T}\right) ds
\end{split}
\end{equation}
with $\Psi_2 = \Psi_{2,+} - \Psi_{2,+,-}$ solution of 
\[
\begin{split}
\Psi_2(t) &=\frac{\mathbb{E}(J^2)}{\mathbb{E}(J)^2}\left(-i\frac{\partial C^+}{\partial a_+}\left(0,0,t\right) - \left(-i\frac{\partial C^-}{\partial a_+}\left(0,0,t\right)\right)\right)^2 + \mathbb{E}(J)\int_0^t \Psi_{2}(s)\varphi_{\exp}(t-s)ds\\
&= \tilde{\Psi}(t) + \mathbb{E}(J)\int_0^t \Psi_{2}(s)\varphi_{\exp}(t-s)ds
\end{split}
\]
with 
\[\tilde{\Psi}(t) = \mathbb{E}(J^2)\left(\frac{\beta}{(\beta + \alpha \mathbb{E}(J))} + \frac{\alpha \mathbb{E}(J)}{(\beta + \alpha \mathbb{E}(J))}e^{-(\beta + \alpha\mathbb{E}(J))t}\right)^2.
\]
Lemma~\ref{lemma:ode} implies that
\[
\begin{split}
\Psi_2(t) =& \tilde{\Psi}(t) + \int_0^t \tilde{\Psi}(s) \mathbb{E}(J)\varphi_{\exp}(t-s)e^{\alpha\mathbb{E}(J)(t-s)}ds \\
=&\mathbb{E}(J^2) \Bigg(\frac{\alpha^2 \mathbb{E}(J)^2}{(\beta + \alpha \mathbb{E}(J))^2}e^{-2(\beta + \alpha \mathbb{E}(J))t} + \frac{\beta^2}{(\beta + \alpha \mathbb{E}(J))^2} + 2\frac{\alpha \beta \mathbb{E}(J)}{(\beta + \alpha \mathbb{E}(J))^2}e^{-(\beta + \alpha \mathbb{E}(J))t} \\
&+ \frac{\alpha^3 \mathbb{E}(J)^3}{(\beta + \alpha \mathbb{E}(J))^2(\beta + 3\alpha \mathbb{E}(J))}\left(e^{-(\beta - \alpha\mathbb{E}(J))t} - e^{-2(\beta + \alpha \mathbb{E}(J))t}\right) \\
&+ \frac{\beta^2\alpha\mathbb{E}(J)}{(\beta+\alpha\mathbb{E}(J))^2(\beta-\alpha\mathbb{E}(J))} \left(1 - e^{-(\beta-\alpha\mathbb{E}(J))t}\right) \\
&+ \frac{\beta \alpha \mathbb{E}(J)}{(\beta+\alpha  \mathbb{E}(J))^2} \left(e^{-(\beta-\alpha \mathbb{E}(J))t} - e^{-(\beta+\alpha \mathbb{E}(J))t}\right) \Bigg)\\
=&\mathbb{E}(J^2)\left(\tilde{C}_{1}e^{-(\beta-\alpha\mathbb{E}(J))t} + \tilde{C}_{2}e^{-2(\beta+\alpha \mathbb{E}(J))t} + \tilde{C}_{3} e^{-(\beta+\alpha \mathbb{E}(J))t} + \tilde{C}_{4}\right)
\end{split}
\]
with 
\begin{align*}
    \tilde{C}_{1} &= \frac{- \alpha^2 \mathbb{E}(J)^2}{(\beta - \alpha \mathbb{E}(J)) (\beta+3\alpha \mathbb{E}(J))},\\
    \tilde{C}_{2} &= \frac{\alpha^2 \mathbb{E}(J)^2(\beta + 2\alpha\mathbb{E}(J))}{(\beta + \alpha\mathbb{E}(J))^2(\beta + 3\alpha\mathbb{E}(J))},\\
    \tilde{C}_{3} &= \frac{\alpha \beta \mathbb{E}(J)}{(\beta+\alpha \mathbb{E}(J))^2},\\
    \tilde{C}_{4} &= \frac{\beta^3}{(\beta+\alpha\mathbb{E}(J))^2(\beta-\alpha\mathbb{E}(J))}.
\end{align*}
Finally, from \eqref{eq:moment2integral}, we obtain
\[
\mathbb{E}(f_t^2) = f_0^2 + 2\mu_0\mathbb{E}(J^2)\left((C_{1} + C_{2} + C_{3} + C_{4})e^{\kappa \frac{t}{T}} - C_{1} e^{-(\beta-\alpha \mathbb{E}(J))t}  - C_{2} e^{-2(\beta+\alpha \mathbb{E}(J))t} - C_{3}e^{-(\beta+\alpha \mathbb{E}(J))t} - C_{4}\right)
\]
with
\begin{align*}
    C_{1} &= \frac{\tilde{C}_{1}}{\beta - \alpha \mathbb{E}(J) + \frac{\kappa}{T}},\\
    C_{2} &= \frac{\tilde{C}_{2}}{2\beta + 2\alpha \mathbb{E}(J) + \frac{\kappa}{T}},\\
    C_{3} &= \frac{\tilde{C}_{3}}{\beta + \alpha \mathbb{E}(J) + \frac{\kappa}{T}},\\
    C_{4} &= \frac{\tilde{C}_{4}}{\frac{\kappa}{T}}.
\end{align*}

\subsection{Proof of Proposition \ref{prop:signatureplot}}
\label{proof:propsignatureplot}
Let $0 \leq s \leq t \leq T$. To compute $\mathbb{E}((f_t - f_s)^2)$, one needs to know $\mathbb{E}(f_t f_s)$ which is equal to the sum of $f_0^2$ and 
\[\mathbb{E}((f_t^+ - f_t^-)(f_s^+ - f_s^-)) = \mathbb{E}(f_t^+f_s^+) + \mathbb{E}(f_t^- f_s^-) - \mathbb{E}(f_t^+ f_s^-) - \mathbb{E}(f_t^- f_s^+).\]
Conditioning with respect to $\mathcal{F}_s$ gives 
\[\mathbb{E}\left(f_t^+f_s^+\right)
= \mathbb{E}\left(\left(f_s^+\right)^2\right) + \mathbb{E}\left(J\right)\mathbb{E}\left(\int_s^t \mathbb{E}\left(\lambda^+_u | \mathcal{F}_s\right)du f_s^+\right),\] 
\[\mathbb{E}\left(f_t^+f_s^-\right)
= \mathbb{E}\left(f_s^+ f_s^-\right) + \mathbb{E}\left(J\right)\mathbb{E}\left(\int_s^t\mathbb{E}\left(\lambda^+_u | \mathcal{F}_s\right)du f_s^-\right),\]
\[\mathbb{E}\left(f_t^- f_s^-\right)
= \mathbb{E}\left(\left(f_s^-\right)^2\right) + \mathbb{E}\left(J\right)\mathbb{E}\left(\int_s^t \mathbb{E}\left(\lambda^-_u | \mathcal{F}_s\right)du f_s^-\right),\] 
\[\mathbb{E}\left(f_t^- f_s^+\right)
= \mathbb{E}\left(f_s^- f_s^+\right) + \mathbb{E}\left(J\right) \mathbb{E}\left(\int_s^t\mathbb{E}\left(\lambda^-_u | \mathcal{F}_s\right)du f_s^+\right).\]
Hence
\[\mathbb{E}\left(f_tf_s\right) = \mathbb{E}\left(f_s^2\right) + \mathbb{E}\left(J\right)\mathbb{E}\left(\int_s^t \mathbb{E}\left(\lambda^+_u - \lambda^-_u |\mathcal{F}_s\right)du \left(f_s^+ - f_s^-\right)\right).\]
The conditional intensities are such that
\[
\begin{split}
\mathbb{E}\left(\lambda^+_u | \mathcal{F}_s\right) &= \mu\left(\frac{u}{T}\right) + \int_0^s \varphi_{\exp}\left(u-v\right)J_vdN^-_v + \mathbb{E}\left(J\right)\int_s^u \varphi_{\exp}\left(u-v\right) \mathbb{E}\left(\lambda^-_v | \mathcal{F}_s\right)dv \\
&= \mu\left(\frac{u}{T}\right) + e^{-\beta\left(u-s\right)} \left(\lambda_s^+ - \mu\left(\frac{s}{T}\right)\right) +  \mathbb{E}(J)\int_s^u \varphi_{\exp}\left(u-v\right) \mathbb{E}\left(\lambda^-_v | \mathcal{F}_s\right)dv
\end{split}
\]
and 
\[
\mathbb{E}\left(\lambda^-_u | \mathcal{F}_s\right) = \mu\left(\frac{u}{T}\right) + e^{-\beta\left(u-s\right)} \left(\lambda_s^- - \mu\left(\frac{s}{T}\right)\right) +  \mathbb{E}(J)\int_s^u \varphi_{\exp}\left(u-v\right) \mathbb{E}\left(\lambda^+_v | \mathcal{F}_s\right)dv.
\]
Applying Lemma \ref{lemma:ode}, we find 
\begin{equation}
\label{diffintensity}
\mathbb{E}\left(\left(\lambda^+_u - \lambda_u^-\right) | \mathcal{F}_s\right) = \left(\lambda_s^+ - \lambda_s^-\right) e^{-\left(\beta+\alpha\mathbb{E}\left(J\right)\right)\left(u-s\right)}.
\end{equation}
As
\[
\mathbb{E}\left(\left(f_s^+\right)^2\right) = \mathbb{E}\left(J^2\right)\mathbb{E}\left(\int_0^s \lambda^+_v dv\right) + 2 \mathbb{E}\left(J\right)\mathbb{E}\left(\int_0^s f^+_v \lambda^+_v dv\right),
\]
\[
\mathbb{E}\left(\left(f_s^-\right)^2\right) = \mathbb{E}\left(J^2\right)\mathbb{E}\left(\int_0^s\lambda^-_v dv\right) + 2 \mathbb{E}\left(J\right)\mathbb{E}\left(\int_0^s f^-_v\lambda^-_v dv\right),
\]
and 
\[
\mathbb{E}\left(f_s^+ f_s^-\right) = 2 \mathbb{E}\left(J\right)\mathbb{E}\left(\int_0^s f^-_v\lambda^+_v dv\right) = 2 \mathbb{E}\left(J\right)\mathbb{E}\left(\int_0^s f^+_v\lambda^-_v dv\right),
\]
we have, using also Equation~\eqref{diffintensity}, 
\[\mathbb{E}(f_tf_s) = \mathbb{E}(f_s^2) + \left(\frac{1}{2}\frac{d \mathbb{E}((f_s^+)^2) + \mathbb{E}((f_s^-)^2) - 2\mathbb{E}(f_s^+ f_s^-)}{ds} - \mathbb{E}(J^2)\mathbb{E}(\lambda_s^+)\right)\frac{(1 - e^{-(\beta+\alpha \mathbb{E}(J))(t-s)})}{\beta + \alpha \mathbb{E}(J)}\]
which is equal to 
\[
\mathbb{E}(f_s^2) + \left(\frac{1}{2} \frac{d \mathbb{E}(f_s^2)}{ds} -\mathbb{E}\left(J^2\right)\mathbb{E}\left(\lambda_s^+\right)\right) \frac{\left(1 - e^{-\left(\beta+\alpha\mathbb{E}\left(J\right)\right)\left(t-s\right)}\right)}{\beta + \alpha\mathbb{E}(J)}.
\]
We obtain 
\[
\mathbb{E}\left(\left(f_t-f_s\right)^2\right) = \mathbb{E}(f_t^2) - \mathbb{E}(f_s^2) -\frac{\left(1 - e^{-\left(\beta+\alpha\mathbb{E}\left(J\right)\right)\left(t-s\right)}\right)}{\beta + \alpha\mathbb{E}(J)} \left(\frac{d \mathbb{E}(f_s^2)}{ds} - 2\mathbb{E}\left(J^2\right)\mathbb{E}\left(\lambda_s^+\right)\right),
\]
which achieves the proof.

\section{Proofs for Section \ref{sec:macro}}
\label{proof:sectionmacro}
\subsection{Proof of Proposition \ref{prop:lln}}
\label{proof:proplln}
The proof is very similar to the one of \cite[Theorem 1]{Bacry2013b}. Let us recall the main steps and identify the differences. In the following, one denotes by $\underline{\mu}$ the function $t \mapsto \mu(t) \left(1, 1\right)^{\top}$.

\medskip
\paragraph{\it Computation of the expectation} First step consists in computing the expectation of $N^0_{t} = N_t$. We easily find that, for $t \in \left[0,T\right]$,
\[
\mathbb{E}(N_t) = \int_0^t \underline{\mu}\left(\frac{s}{T}\right) ds + \int_0^t \Psi(t-s) \left(\int_0^s \underline{\mu}\left(\frac{u}{T}\right) du\right)ds
\]
where $\Psi$ is defined by $\sum_{n \geq 1} \varphi_n$ and $\varphi_n$ is defined recursively by
\[
    \begin{cases}
        \varphi_1 &= \mathbb{E}(J) \varphi \\
        \varphi_{n+1}(t) &= \int_0^t \mathbb{E}(J) \varphi(t-s)\varphi_n(s)ds, \; n \geq 1.
    \end{cases}
\]
The sum is well defined: as $\int_0^{\infty} \varphi_n(s) ds = K^n$ and $\rho(K) < 1$, $\sum_{n \geq 1} \varphi_n$ converges in $L^1(dt)$. Furthermore, $\int_0^{\infty} \Psi(s)ds = (I_2 - K)^{-1} - I_2$.

\medskip
By doing a simple change of variable in the integral, we have for $v \in \left[0,1\right]$,
\begin{equation} \label{eq:expectrescale}
\begin{split}
\frac{\mathbb{E}(N_{Tv})}{T} &= \int_0^v \underline{\mu}(s) ds + \frac{1}{T}\int_0^{Tv} \Psi(Tv-s) \left(\int_0^s \underline{\mu}\left(\frac{u}{T}\right) du\right)ds \\
&= \int_0^v \underline{\mu}(s) ds + \frac{1}{T}\int_0^{Tv} \Psi(s) \left(\int_0^{Tv-s} \underline{\mu}\left(\frac{u}{T}\right) du\right)ds \\
&= \left(I_2 + \int_{0}^{Tv} \Psi(s) ds\right)\int_0^{v} \underline{\mu}(s) ds - \frac{1}{T} \int_0^{Tv} \Psi(s) \left(\int_{0}^s \underline{\mu}\left(v - \frac{u}{T}\right) du\right) ds.
\end{split}
\end{equation}

\medskip
\paragraph{\it Approximation of $\mathbb{E}\left(\frac{N^i_{Tv}}{T}\right)$} From Equation~\eqref{eq:expectrescale}, 
\[\left(I_2 - K\right)^{-1}\int_0^{v} \underline{\mu}(s) ds - \frac{\mathbb{E}(N_{Tv})}{T} = \int_{Tv}^{\infty} \Psi(s)ds \int_0^v \underline{\mu}(s) ds + \frac{1}{T} \int_0^{Tv} \Psi(s)\left(\int_0^s  \underline{\mu}\left(v - \frac{u}{T}\right) du\right) ds. 
\]
As $\mu$ is bounded, one has componentwise
\begin{equation} \label{eq:approxexpectationbound}
\left(I_2 - K\right)^{-1}\int_0^{v} \underline{\mu}(s) ds - \frac{\mathbb{E}(N_{Tv})}{T} \leq \left(v\int_{Tv}^{\infty} \Psi(s)ds  + \frac{1}{T} \int_0^{Tv} \Psi(s) s ds\right)\underline{\mu}^{*}
\end{equation}
with $\underline{\mu}^{*}$ an upper bound for $\underline{\mu}$ componentwise on $\left[0,1\right]$. The right hand side of \eqref{eq:approxexpectationbound} is equal to $v \left(I_2 - K\right)^{-1}\underline{\mu}^{*} - \frac{\mathbb{E}(\tilde{N}_{Tv})}{T}$ where $\tilde{N}$ is a Hawkes process with constant baseline intensity $\underline{\mu}^{*}$ and excitation kernel $\mathbb{E}(J) \varphi$. By \cite[Lemma 5]{Bacry2013b}, it converges to 0 uniformly on $\left[0,1\right]$ as $T \to \infty$, so that we have 
\begin{equation} \label{eq:convexpect}
\frac{\mathbb{E}(N_{Tv})}{T} - (I_2 - K)^{-1}\int_0^v \mu(s) ds \to 0
\end{equation}
uniformly on $\left[0,1\right]$ as $T \to \infty$.
As $\mathbb{E}(N^i_{t}) = \mathbb{E}(J^i)\mathbb{E}(N_{t})$, we obtain directly that $\frac{\mathbb{E}(N^i_{Tv})}{T}$ converges uniformly on $\left[0,1\right]$ as $T \to \infty$ to $\mathbb{E}(J^i)(I_2 - K)^{-1} \int_0^v \mu(s) ds$.

\medskip
\paragraph{\it Approximation of $\frac{N^i_{Tv}}{T}$}
Let us consider the martingales $M^i_t = N^i_t - \mathbb{E}(J^i)\int_0^t \lambda_s ds$ for $t \in \left[0,T\right]$ and $i \in \{0,1,2\}$. As 
\[
\mathbb{E}(N^1_t) = \mathbb{E}(J)\int_0^t \underline{\mu}\left(\frac{s}{T}\right) ds + \mathbb{E}(J)\int_0^t \left(\int_0^s  \varphi(s-u) \mathbb{E}(J_udN_u)\right)ds
\]
and 
\[
\mathbb{E}(J) \int_0^t \lambda_s ds = \mathbb{E}(J)\int_0^t \underline{\mu}\left(\frac{s}{T}\right)ds + \mathbb{E}(J)\int_0^t \left(\int_0^s \varphi(s-u)J_u dN_u\right)ds,
\]
one gets
\[
N^1_t - \mathbb{E}(N^1_t) = M^1_t + \mathbb{E}(J)\int_0^t \left(\int_0^s  \varphi(s-u)\left(dN^1_u - \mathbb{E}(dN^1_u)\right)\right)ds.
\]
Then, applying Fubini theorem, we have
\[
    N^1_t - \mathbb{E}(N^1_t) = M^1_t + \mathbb{E}(J)\int_0^t \left(\int_u^t  \varphi(s-u)ds\right)\left(dN^1_u - \mathbb{E}(dN^1_u)\right)
\]
so that using integration by parts, similarly to the proof of \cite[Lemma 2]{Bacry2013b}, one finds
\begin{equation} \label{eq:relN1}
N^1_t - \mathbb{E}(N^1_t) = M^1_t + \mathbb{E}(J)\int_0^t \varphi(t-s)\left(N^1_s - \mathbb{E}(N^1_s)\right)ds.
\end{equation}
Applying \cite[Lemma 3]{Bacry2013b}, we obtain 
\begin{equation} \label{eq:relNM}
N^1_t - \mathbb{E}(N^1_t) = M^1_t + \int_0^t \Psi(t-s)M^1_s ds.
\end{equation}
Let us now consider the case $i \in \{0, 2\}$. We have 
\[
\mathbb{E}(N^i_t) = \mathbb{E}(J^i) \int_0^t \underline{\mu}\left(\frac{s}{T}\right) ds + \mathbb{E}(J^i)\int_0^t \int_0^s \varphi(s-u) \mathbb{E}(dN^1_u)
\]
and 
\[
\mathbb{E}(J^i)\int_0^t \lambda_s ds = \mathbb{E}(J^i) \int_0^t \underline{\mu}\left(\frac{s}{T}\right) ds + \mathbb{E}(J^i)\int_0^t \int_0^s \varphi(s-u)dN^1_u.
\] 
Thus, integrating by parts again, 
\[
N_t^i - \mathbb{E}(N^i_t) = M_t^i + \mathbb{E}(J^i) \int_0^t \varphi(t-s) \left(N_s^1 - \mathbb{E}(N_s^1)\right)ds
\]
and we have, using \eqref{eq:relN1} and \eqref{eq:relNM}, 
\[
N_t^i - \mathbb{E}(N^i_t) = M_t^i + \frac{\mathbb{E}(J^i)}{\mathbb{E}(J)} \int_0^t \Psi(t-s) M^1_s ds.
\]
Remaining of the proof is the same as the one of \cite[Theorem 1]{Bacry2013b} :
\begin{itemize}
    \item $\sup_{v \in \left[0,1\right]} T^{-1}\|N_{Tv}^i - \mathbb{E}(N^i_{Tv})\|$ is bounded by $$T^{-1}\sup_{t \in \left[0,T\right]} \|M^i_t\| + \frac{\mathbb{E}(J^i)}{\mathbb{E}(J)} \int_0^{\infty} \|\Psi(s)ds\|  T^{-1}\sup_{t \in \left[0,T\right]} \|M^1_t\|;$$
    \item if $\mathbb{E}(J^{2i}) < \infty$, we can use Doob's inequality to bound $\mathbb{E}((\sup_{t \in \left[0,T\right]} \|M^i_t\|)^2)$ by a constant times $T$ (the quadratic variation of $M^i$ being $N^{2i}$ on the diagonal and 0 otherwise, and we can easily show that $\mathbb{E}(N^{2i}_T)$ is bounded by $T$ as $\underline{\mu}$ is bounded) and the convergence of $\sup_{v \in \left[0,1\right]} T^{-1}\|N_{Tv}^i - \mathbb{E}(N_{Tv})\|$ toward 0 in $L^2(\mathbb{P})$ follows ;
    \item the arguments for the almost-sure convergence remain true.
\end{itemize}
We then have \[ \sup_{v \in \left[0,1\right]} \|N^i_{Tv} - \mathbb{E}(N^i_{Tv})\| \to 0\]
for $i \in \{0,1,2\}$ almost-surely and in $L^2(\mathbb{P})$ as $T \to \infty$, ending the proof.

\subsection{Proof of Proposition \ref{prop:brownianlimit}}
\label{proof:propbrownianlimit}

The proof relies on the convergence of the martingale 
\[M^{1,(T)} =  \left(T^{-\frac{1}{2}} \left(N^1_{Tv} - \mathbb{E}(J)\int_0^{Tv} \lambda_s ds\right)\right)_{v \in \left[0,1\right]}.
\]
The quadratic variation matrix is equal on the diagonal to $\frac{1}{T} N^{2}_{Tv}$ and 0 otherwise, and then converges to $C_{v} = \mathbb{E}(J^2) \Sigma \int_0^v \mu(s)ds$, with $\Sigma_{i,i} = \left((I_2 - K)^{-1} (1,1)^{\top}\right)_i$, $i=1,2$, in $L^2(\mathbb{P})$ according to Proposition \ref{prop:lln}.

\medskip
The jump measure compensator of $M^{1,(T)}$ is $\nu^{(T)}(dt,dx) = T\lambda_{tT}dt \otimes \mu_J(dx\sqrt{T})$ where $\mu_J$ is the probability measure associated to $J$. Let $t \in \left[0,1\right]$, $\epsilon > 0$. The integral
\[
\int_0^t \int_{\mathbb{R}_+} x^2 {\bf 1}_{x > \epsilon} \nu^{(T)}(ds,dx)
\]
is equal to 
\[\mathbb{E}(J^2 {\bf 1}_{J > \sqrt{T}\epsilon})\int_0^t \lambda_{sT} ds\] 
with expectation equal to 
\[\mathbb{E}(J^2 {\bf 1}_{J > \sqrt{T}\epsilon})\frac{\mathbb{E}(N_{Tt})}{T}\]
which is bounded componentwise by \[\mathbb{E}(J^4)T^{-1} \epsilon^{-2} \frac{\mathbb{E}(N_{Tt})}{T}\]
as can be seen by applying Cauchy-Schwarz inequality to $\E(J^2{\bf 1}_{J> \sqrt{T}\epsilon})$ and then Markov inequality $\mathbb{P}(J>\sqrt{T}\epsilon)\leq \frac{\E(J^4)}{T^2\epsilon^4}$. The bound converges to $0$ when $T \to \infty$ using \eqref{eq:convexpect} from the proof of Proposition \ref{prop:lln}. Then $\int_0^t \int_{\mathbb{R}_+} x^2 {\bf 1}_{x > \epsilon} \nu^{(T)}(ds,dx)$ converges in probability to $0$ when $T \to \infty$ (as it is positive) for every $t \in \left[0,1\right]$, $\epsilon > 0$. From \cite[Section VIII, Theorem 3.22]{Jacod13}, the quadratic variation convergence implies the convergence in law of $M^{1,(T)}$ to a martingale with quadratic variation $C$, that is to $\left(\sqrt{\mathbb{E}(J^2)} \Sigma^{1/2} \int_0^{v} \sqrt{\mu(s)}dW_s\right)_{v \in \left[0,1\right]}$ with $W$ a 2-dimensional Brownian motion.

\medskip
Remaining of the proof is similar to the one of \cite[Theorem 2]{Bacry2013b}

\vip 
\noindent {\bf Acknowledgements.} This research is supported by the FiME (Finance for Energy Markets) Research Initiative.

\bibliographystyle{plain}
\bibliography{Biblio}
\end{document}